\documentclass{article}
\usepackage{amsmath,amssymb,amsthm,color,stmaryrd,graphicx,subcaption}
\usepackage[font=small,labelfont=bf]{caption}
\usepackage{hyperref}
\hypersetup{colorlinks=true}
\usepackage{cite}
\usepackage[utf8]{inputenc}
\usepackage{color}

\makeatletter
\renewcommand\@biblabel[1]{#1.}
\makeatother

\theoremstyle{plain}
\newtheorem{theorem}{Theorem}
\newtheorem{proposition}[theorem]{Proposition}
\newtheorem{corollary}[theorem]{Corollary}

\newtheorem{definition}[theorem]{Definition}

\theoremstyle{remark}
\newtheorem{remark}{Remark}

\def\d{{\rm d}}
\newcommand{\pdx}[1]{\frac{\partial}{\partial x_{#1}}}
\def\A{\mathcal{A}}
\def\E{\mathcal{E}}
\def\P{\mathcal{P}}
\def\D{\mathcal{D}}
\def\O{\mathcal{O}}

\def\Ll{\llbracket}
\def\Lr{\rrbracket}
\def\LBr{\Ll \cdot, \cdot \Lr}

\numberwithin{equation}{section}
\numberwithin{theorem}{section}
\input xy
\xyoption{all}

\voffset = - 1.2cm
\begin{document}

\centerline{{\Large \bf Tensorial dynamics on the space of quantum states}}
\vskip 0.5cm
\centerline{J.F. Cari\~nena$^{a)}$, J. Clemente-Gallardo$^{a,b)}$, J.A. Jover-Galtier$^{a,b)}$, G. Marmo$^{c,d)}$.}
\vskip 0.25cm
\centerline{$^{a)}$Departamento de F\'isica Te\'orica, Facultad de Ciencias, Universidad de Zaragoza.}
\centerline{c. Pedro Cerbuna 12, 50.009 Zaragoza, Spain.}
\centerline{$^{b)}$Instituto de Biocomputaci\'on y F\'isica de Sistemas Complejos (BIFI).}
\centerline{c. Mariano Esquillor (Edificio I+D), 50.018 Zaragoza, Spain.}
\centerline{$^{c)}$Dipartimento di Fisica, Universit\`a degli Studi di Napoli Federico II.}
\centerline{Via Cintia, 80.126 Napoli, Italy.}
\centerline{$^{d)}$INFN, Sezione di Napoli. Via Cintia, 80.126 Napoli, Italy.}
\vskip 1cm

{\bf Abstract:} A geometric description of the space of states of a finite-dimensional quantum system and of the Markovian evolution associated with the Kossakowski-Lindblad operator is presented. This geometric setting is based on two composition laws on the space of observables defined by a pair of contravariant tensor fields. The first one is a Poisson tensor field that encodes the commutator product and allows us to develop a Hamiltonian mechanics. The other tensor field is symmetric, encodes the Jordan product and provides the variances and covariances of measures associated with the observables. This tensorial formulation of quantum systems is able to describe, in a natural way, the Markovian dynamical evolution as a vector field on the space of states. Therefore, it is possible to consider dynamical effects on non-linear physical quantities, such as entropies, purity and concurrence. In particular, in this work the tensorial formulation is used to consider the dynamical evolution of the symmetric and skew-symmetric tensors and to read off the corresponding limits as giving rise to a contraction of the initial Jordan and Lie products.

\vspace{0.5cm}




\section{Introduction}

In the last forty years there has been a widespread interest in a geometrical description of quantum phenomena analogous to the Hamiltonian description of classical mechanics. Most of the approaches have been done in the the Schr\"odinger picture of pure states, or equivalently, of the dynamics of quantum pure states described by Landau-von Neumann equation \cite{Anandan1991, Anandan1990, Aniello2010, Aniello2011, Boya1991, Clemente-Gallardo2012, Clemente-Gallardo2007, Clemente-Gallardo2008, Clemente-Gallardo2009, Ercolessi2010, Kibble1979, Strocchi1966}. In this setting, the complex Hilbert space $\mathcal{H}$ associated to a quantum system is identified with a real differentiable manifold $M_Q$ endowed with a complex structure, i.e. a $(1,1)$-tensor field that reproduces the complex linearity properties of $\mathcal{H}$. There also exists a canonical K\"ahler structure induced by the Hermitian {inner product defining the Hilbert space structure, i.e. $M_Q$ is a K\"ahler manifold. Similar considerations, using the natural projection, prove that the projective Hilbert space $\P$ is also a K\"ahler manifold. Thus, the space of pure states of quantum systems is endowed with a geometric structure.

The other ingredients of a quantum system have to be also described in geometric terms: observables are represented by some functions on the manifold $\P$, whose values at each point are precisely the expectation values of the observables in the corresponding state of the system. The Schr\"odinger equation, which governs the unitary evolution of isolated systems, is represented by a vector field which is both a Hamiltonian and a Killing vector field with respect to the canonical K\"ahler structure on the complex projective space. For a complete presentation of this geometrical description, see the aforementioned references.

This study of dynamics of pure states in the Schr\"odinger picture is adequate to describe isolated physical systems. However, more realistic evolutions associated with open systems require a statistical analysis of dynamics. From this point of view, the states of quantum systems are described by density matrices \cite{Breuer2002, Hall2013}, which are convex combinations of projectors on one-dimensional subspaces of the Hilbert space. The state of the system is said to be pure if it corresponds to one of those projectors (or equivalently, a point in the complex projective space). Otherwise the state of the system is said to be mixed. Their corresponding density matrices are the required tool to understand the general behaviour of open quantum systems, i.e. systems whose evolution is influenced by an external environment whose properties can only, at most, be averaged. 

Following the description of Hilbert and projective complex spaces in geometrical terms, this work aims to characterise geometrically the pure and mixed states of quantum systems, by extending some previous work done by some of us \cite{Carinena2007a, Clemente-Gallardo2007, Clemente-Gallardo2013, Clemente-Gallardo2013a}, where a tensorial description of the set of observables $\mathcal{O}$ (and its dual space $\mathcal{O}^{*}$) was considered. Density matrices can be identified as $\mathbb{C}$-linear functionals acting on a complex unital C*-algebra $\A$, whose real elements represent the physical observables in the Heisenberg picture. These real elements form a real Lie-Jordan algebra $\O$. States are described by the subset of normalised, positive functionals of the dual space $\O^*$. It has been shown \cite{Grabowski2005,Grabowski2006} that, from the geometrical point of view, the set of states is a stratified manifold with a boundary, and a convex set whose extremal points are the pure states.

The geometric description of the dual space $\O^*$ leads to the definition of a pair of tensor fields on $\O^*$, a symmetric one, $R$ (of Jordan type, in this case), and a skew-symmetric one, $\Lambda$ (the canonical Lie-Poisson structure on the dual of the Lie algebra). Similarly to the case of the Schrödinger picture, these tensor fields encode the relevant algebraic structures of the algebra of observables. In the present work, we will describe how those tensors of $\O^*$ can be used to obtain another pair of tensors on the subset of states which will allow us to encode the Heisenberg formalism in tensorial language. Geometric tools can be thus employed in order to characterise the evolution of quantum systems. Among the advantages of the formalism, it is possible to study the change in non-linear properties, such as entropies or purity, or even more abstract ones, as the algebraic properties of quantum observables. This last feature will be dealt with in detail for the case Markovian systems.

Markovian dynamical maps were classified in the mid seventies by Gorini, Kossakowski, Sudarshan and coworkers \cite{Gorini1976a, Gorini1978, Kossakowski1972}, and by Lindblad \cite{Lindblad1976}. The non-unitary evolution defined by the Kossakowski-Lindblad dynamics has some novel properties not present in the case of unitary evolutions. Some characteristics of the Kossakowski-Lindblad evolution, such as the existence and properties of limit states, have been studied by Zanardi \cite{Zanardi1998}, Baumgartner and coworkers \cite{Baumgartner2008b, Baumgartner2008a}, Albert and coworkers \cite{Albert2015} and many others (see references in those works). A complete geometric description of their properties is given in this paper. Thus, the vector field corresponding to the system of differential equations that determines the evolution is described by the sum of three vector fields: a Hamiltonian one, a gradient one and a third vector field associated with a Kraus-type action on the space of states. The definition of each one will become clear in the text. This geometrical characterisation allows us to consider the dynamics from a new perspective.

As an application of the method, we will reproduce in terms of tensor
fields the contraction of algebras of observables presented in
\cite{Alipour2015, Chruscinski2012, Ibort2016}. The theory of contractions of algebras was introduced in the 1950's by I. E. Segal \cite{Segal1951}, E. In\"on\"u and E. P. Wigner \cite{Inonu1953}, and later extended by other authors \cite{Carinena2000, Carinena2001, Carinena2004, Saletan1961, Weimar-Woods1991, Weimar-Woods1991a, Weimar-Woods2000, Weimar-Woods2006}. By a well-defined procedure that will be detailed below, it is possible to obtain new algebras from a given one by considering the asymptotic limit of a family of linear transformations of the algebra. This new algebra is called a contraction of the initial one. In the context of Quantum Mechanics, some types of evolutions may define contractions of the Lie-Jordan algebra of observables of the system. This is a relatively unknown topic that, nevertheless, may have interesting applications that will be studied in future works. The geometric formalism presented here is particularly well-suited for the description of contractions, since it allows to formulate them at the level of the observables and, at the same time, at the level of the quantum states. The results presented here may shed some light on the relevance of these contractions in the study of quantum systems.

The paper is organised as follows. Section \ref{secGeom} summarises
some of the results presented in previous works, and focuses on the
new tensorial approach to the dynamics on the space of states. Hence,
instead of describing the dynamics on the whole dual space of the
Lie-Jordan algebra of observables, we define tensor fields only on the
stratified submanifold of states. Similarly, observables are no longer
represented by linear functions, but by expectation value
functions. The resulting geometrical objects are indeed physically
relevant. The new functions are the correct way to represent
observables on quantum systems, as their values are the results that
can be obtained in measurements of quantum systems. On the other hand,
the resulting tensor fields reproduce the properties of quantum
systems. The new Poisson tensor fields is connected, as before, with
the unitary evolution of quantum systems, while the symmetric one
reproduces geometrically another important aspect of quantum
systems: {\color{black} the variance and covariance of
observables.
}
As an application of these results, in Section \ref{section2levels} a 2-level system is studied. The Hamiltonian and gradient vector fields are computed, and their integral curves are plotted to get some insight on the dynamics of states. In Section \ref{secKL}, we consider the new tensors in order to describe, in geometrical terms, the Kossakowski-Lindblad evolution. Section \ref{secContraction} presents the second main contribution of this paper: we consider Kossakowski-Lindblad evolution and prove that, in simple cases, we can consider its effect on the tensors and capture in this way a contraction of the algebra of observables of the quantum system. The last section of the paper, Section \ref{secConcl}, is devoted to discuss the results presented in the paper and its relevance in the geometrical description of the dynamics of quantum systems. 

\section{The geometry of the set of quantum states}
\label{secGeom}

This section is concerned with the geometric description of the Heisenberg picture of Quantum Mechanics, which characterises physical quantum systems in a purely algebraic setting. In its modern formulation, the Heisenberg picture characterises the observables of a physical system as the real elements of a unital C*-algebra \cite{Carinena2015book, Ercolessi2010, Falceto2013, Falceto2013a, Falceto2013b, Haag1992, Haag1964}. States of the system are normalised positive real $\mathbb{C}$-linear functionals on the C*-algebra. The evaluation of these functionals on observables represents the outcome of a measurement process. 

It is possible to describe the Heisenberg picture in geometric terms
\cite{Carinena2015book, Grabowski2005, Grabowski2006}. This follows
the path initiated by the aforementioned geometric description of the
Schr\"odinger picture. Algebraic properties of the space of
observables naturally define tensor fields on its dual space $\O^*$
Observables themselves are represented by linear functions on this
dual space. Thus, a geometric description of the algebra of
observables is achieved  {\color{black} by identifying the real vector space of observables and
  by introducing two product structures, the Lie and the Jordan
  brackets. We are obliged to do so because the associative product on
the set of (complex) operators, when restricted to real elements
is no longer an inner product. It breaks up into two different inner
products, a skew symmetric one and a symmetric one, which define,
respectively, a Lie product and a Jordan product.}

However, a step further is required in order to truly obtain a
geometric description of the Heisenberg picture. Not every element in
$\O^*$ represents a state of the quantum system. Thus, a geometric
characterisation of the subset $\D$ of states has to be done. As
argued below, such a characterisation cannot be trivially related with
the one obtained for $\O^*$. Instead, the relevant geometrical objects
have to be redefined in order to admit a restriction to $\D$. In
particular, observables are no longer represented by linear functions,
but by expectation value functions. In this way, a proper description
of the geometrical properties of the set of states is obtained. 

A comment should be made regarding the dimensionality of quantum systems. In our analysis, we will restrict ourselves to finite-dimensional quantum systems, and hence to differentiable manifolds. The
extension to general infinite-dimensional cases introduces so many topological subtleties that the new insight given by the geometrical encoding is obscured by the technical problems. Nonetheless, finite-dimensional quantum systems provide useful models to describe a huge number of physical problems, either exactly (such as spin systems and problems of quantum information) or by approximations (as by setting a limit to the energy of quantum systems with unbounded Hamiltonian). In any case, having in mind the possibility of developing a geometrical description of infinite-dimensional quantum systems in future works, statements will be presented, if possible, in a coordinate independent manner that can cover both cases of finite and infinite dimensions.

\subsection{The Lie-Jordan algebra of observables}

\begin{definition}
A C*-algebra $\mathcal{A}$ is \cite[p. 70]{Emch1972} an involutive Banach (associative) algebra over the complex field $\mathbb{C}$, satisfying the condition
\begin{equation}
	\lVert a^ * a \rVert = \lVert a \rVert^2, \quad \forall a \in \mathcal{A},
\end{equation}
where $\lVert \cdot \rVert$ denotes the norm and $ *$ is an involution in $\A$ which is also an antiautomorphism, i.e. an invertible antilinear map such that $(a^*)^* = a$ and $(ab)^* = b^* a^*$. 
An element $a \in \mathcal{A}$ is said to be real if $a^ * = a$. The $\mathbb{R}$-linear subspace of real elements of $\A$ will be denoted $\O$.
\end{definition}

In the Heisenberg picture, one considers the C*-algebra of bounded operators on a Hilbert space, whose real elements are identified as the observables of the quantum system \cite{Strocchi2005}. Real elements do not constitute an associative subalgebra, because the product $ab$ of two real elements is not, in general, real. If the symmetrised product $ab+ba$ is considered instead, then a real element is obtained out of two real elements. Associativity, however, is lost. It is thus possible to introduce some appropriate new algebraic structures, which can be defined by using the associative product.

\begin{definition}
\label{defLJAlg}
Let $\LBr$ be a Lie product on a linear space $V$ over a field $\mathbb{K}$ (either $\mathbb{R}$ or $\mathbb{C}$), that is, a skew-symmetric $\mathbb{K}$-bilinear map from $V\times V$ to $V$ that satisfies Jacobi identity,
\begin{equation}
	\Ll a,\Ll b,c\Lr\Lr + \Ll b,\Ll c,a\Lr\Lr + \Ll c,\Ll a,b\Lr\Lr = 0, \quad \forall a,b,c \in V.
\end{equation}
Let $\odot$ be a Jordan product on $V$, that is, a symmetric $\mathbb{K}$-bilinear product satisfying Jordan identity,
\begin{equation}
	(a \odot b) \odot (a \odot a) = a \odot (b \odot (a \odot a )), \quad \forall a,b,c \in V.
\end{equation}
The pair $(V, \LBr)$ is called a Lie algebra and the pair $(V, \odot)$ is called a Jordan algebra.
The triple $(V, \odot, \LBr)$ is called a Lie-Jordan algebra if, for each $a\in V$, $\Ll a, \cdot \Lr$ is a derivation of the Jordan algebra $(V, \odot)$, i.e. $\Ll a, b\odot c \Lr = \Ll a,b \Lr \odot c + b \odot \Ll a,c \Lr$, and the associators of the Lie product and the Jordan product are proportional:
\begin{equation}
\label{LieJordanRel}
	a\odot (b\odot c) - (a \odot b) \odot c = \mu^2 (\Ll a, \Ll b,c \Lr \Lr - \Ll \Ll a,b\Lr ,c\Lr) ,
\end{equation}
for any $a,b,c \in V$ and some real non-zero number $\mu \in \mathbb{R} - \{0\}$.
\end{definition}

\begin{proposition}
Let $\mathcal{A}$ be an $n$-dimensional C*-algebra with unity, and consider the following products defined in terms of the associative composition law in $\mathcal{A}$:
\begin{equation}
	\Ll a,b \Lr = - \frac i2 (ab - ba), \quad
	a \odot b = \frac 12 (ab+ba), \quad
	a,b \in \O.
\end{equation}
The product $\LBr$ is a Lie bracket, while $\odot$ is a Jordan product. The triple $(\mathcal{O}, \odot, \LBr)$, is a Lie-Jordan algebra over the field of real numbers $\mathbb{R}$.
\end{proposition}
\begin{proof}
The set $\O$ inherits the $\mathbb{R}$-linear structure of $\mathcal{A}$, and is a closed set under both products. The triple $(\mathcal{O}, \odot, \LBr)$ satisfies the axioms of Lie-Jordan algebras, which can be checked by direct computation.
\end{proof}

\begin{remark}
The associative composition law in the C*-algebra is recovered as
\begin{equation}
  \label{eq:7}
ab = a \odot b + i \Ll a,b \Lr.
\end{equation}
In particular, $a^2 = a \odot a$.
\end{remark}

Further details in the description of Lie-Jordan algebras and their relevance in Quantum Mechanics can be found in books by G. Emch \cite{Emch1972} and N. P. Landsman \cite{Landsman1998}.

\subsection{Tensorial description of the dual space of a Lie-Jordan algebra}

It is possible to obtain a tensorial description of the algebraic properties of $\O$ \cite{Carinena2015book, Ercolessi2010, Grabowski2005, Grabowski2006}. This is achieved by considering the natural differentiable structure on the dual space $\O^*$ of real $\mathbb{R}$-linear functionals on $\O$. Several geometrical objects on $\O$ are related with relevant structures in the description of Quantum Mechanics. In particular, quantum observables are identified with linear functions on $\O^*$.
\begin{definition}
Given an element $a \in \O$, we denote by $f_a:\O^*\to\mathbb{R}$ the evaluation function, i.e. the $\mathbb{R}$-linear function on $\O^*$ defined by
\begin{equation}
	f_a (\xi) = \xi (a), \quad \forall \xi \in \O^*.
\end{equation}
\end{definition}
This is an injective $\mathbb{R}$-linear homomorphism of linear spaces, $a\in\O \mapsto f_a\in (O^*)^*$, which is an isomorphism for the finite-dimensional case.

{\color{black}Up to here we have replaced real elements of the
  $C^{*}$--algebra (i.e., the physical observables) with linear
  functions on the dual space $\mathcal{O}^{*}$.  We will denote as
  $\mathcal{F}$ the set of these functions.  Clearly the pointwise product of two
  linear functions will not be linear and therefore the corresponding
  structure will not be an inner product (i.e., the pointwise product
  of two evaluation functions is not an evaluation function). It is however possible to induce a
Lie-Jordan structure on the space of linear functions. Being functions on a
  manifold we can also consider the corresponding differentials and
  their products and provide a tensorial realization of the algebraic
  structures of $\mathcal{O}$. Our experience with Lie algebras and
  their descriptions in terms of Poisson structures on the dual space
  suggests to define a pair of tensors, one symmetric and one
  skew-symmetric  (and their corresponding brackets) to describe the
  Lie-Jordan algebra. }

The linear structure of a vector space implies that both its tangent and cotangent bundles are trivialisable. In the particular case of the dual space $\O^*$, we have $T\O^* = \O^* \times \O^*$ and $T^* \O^* = \O^* \times (\O^*)^* \cong \O^* \times \O$. Unless necessary to distinguish, the same notation will be used for tangent vectors to $\O^*$ and points in $\O^*$, on one side, and for covectors on $\O^*$ and elements in $\O$, on the other. With this trivialisation, the Lie bracket and Jordan product on $\O$ can also be given a geometrical description. They are represented by $(2,0)$-tensor fields on $\O^*$, defined as follows.

\begin{definition}
For any point $\xi \in \O^*$,
let $\lambda_\xi, r_\xi: T^*_\xi \O^* \times T^*_\xi \O^* \rightarrow \mathbb{R}$ be the two $\mathbb{R}$-bilinear maps defined, respectively, by the Jordan product and the Lie product on $\O$ as follows
\begin{equation}
\label{defTrlambda}
	\lambda_\xi (a, b) = \xi (\Ll a, b\Lr) = f_{\Ll a ,b \Lr} (\xi), \quad
	r_\xi (a, b) = \xi (a \odot b) = f_{a \odot b} (\xi),
\end{equation}
for any pair $a,b \in T_\xi^* \O^* = \O$.
\end{definition}

\begin{proposition}
\label{propTens}
Let $\Lambda$ and $R$ denote the sections of $T^{(2,0)} \O^*$ (the twice contravariant tensor bundle of $\O^*$) such that $\Lambda (\xi) = \lambda_\xi$, $R(\xi) = r_\xi$, for any point $\xi \in \O^*$. 
Then, $\Lambda$ and $R$ are smooth
$(2,0)$- tensor fields.
\end{proposition}
\begin{proof}
There are two equivalent ways of understanding smooth tensor fields \cite{Carinena2015book}: either as smooth sections of tensor bundles, or as $C^\infty (\O^*)$-linear maps on vector fields and 1-forms. Both $\Lambda$ and $R$ are, by definition, smooth sections of $T^{(2,0)} \O^*$. Their $C^\infty (\O^*)$-linearity is derived from the $\mathbb{R}$-linearity of the Lie bracket and the Jordan product, thus satisfying the requirements.
\end{proof}

The tensor field $\Lambda$ is the canonical Kirillov-Kostant-Souriau Poisson tensor field, while $R$ is a symmetric tensor field \cite{Grabowski2005}. To compute their coordinate expressions, 
let us fix the notations 
for elements in the Lie-Jordan algebra $\O$ and in its dual space
$\O^*$. Assume that $\O$ carries an inner {\color{black} scalar}
product, and consider an orthonormal basis $\{\sigma_j\}_{j=1}^n$ for
$\O$. An observable $a \in \O$ takes the form $a = a^j \sigma_j$,  
with $a^1, \ldots, a^n \in \mathbb{R}$ and where summation over repeated indices is understood. The composition laws in $\O$ are determined by their structure constants $c_{jk}^l$ and $d_{jk}^l$:
\begin{equation}
\label{strConst}
\Ll \sigma_j, \sigma_k \Lr = c_{jk}^l \sigma_l, \quad
\sigma_j \odot \sigma_k = d_{jk}^l \ \sigma_l, \quad
j,k = 1,2, \ldots , n,
\end{equation}
where $c_{jk}^l = -c_{kj}^l$ and $d_{jk}^l = d_{kj}^l$. Let $\{\sigma^j\}_{j=1}^n$ be the dual basis on $\O^*$, i.e. the set of linear functions on $\O$ satisfying
\begin{equation}
\sigma^j (\sigma_k) = \delta_k^j, \quad j,k = 1,2, \ldots, n.
\end{equation}
Any element $\xi \in \O^*$ can be decomposed in this basis as $\xi = \xi_j \sigma^j$. Coordinate functions on $\O^*$ with respect to the given basis are functions associated to the elements in the basis of $\O^*$:
\begin{equation}
\xi_j = \xi (\sigma_j) = f_{\sigma_j} (\xi), \quad \xi \in \O.
\end{equation}
Coordinate functions on $\O^*$ will be denoted as $x_j = f_{\sigma_j}$. The function associated to an observable $a = a^j \sigma_j$ is thus $f_a = a^j x_j$. With this notation, 
and in view of \eqref{defTrlambda} and \eqref{strConst}, the coordinate expressions of the tensor fields $\Lambda$ and $R$ are the following:
\begin{equation}
\label{LRcoord}
\Lambda = \frac 12 \ c_{jk}^l x_l \pdx{j} \wedge \pdx{k}, \quad
R = d_{jk}^l x_l \pdx{j} \otimes \pdx{k}.
\end{equation}
where $v \wedge w = v \otimes w - w \otimes v$, following the definition for the exterior product of M. Crampin and F. A. E. Pirani \cite{Crampin1986}.

The properties of $\Lambda$ and $R$ as tensor fields allow us to associate a Hamiltonian and a gradient vector fields to any smooth function on $\O^*$ \cite{Grabowski2006}. These vector fields act, as usual, as derivations on the algebra of smooth functions with respect to the usual point-wise product.

\begin{definition}
\label{dfnXfYf}
Let $X_f$ and $Y_f$ denote the Hamiltonian and gradient vector fields, respectively, on $\O^*$ associated with a smooth function $f \in C^\infty(\O^*)$ by means of $\Lambda$ and $R$:
\begin{equation}
	X_f = \iota (\d f) \Lambda, \quad
	Y_f = \iota (\d f) R, \quad
	f \in C^\infty (\O^*).
\end{equation}
where $\iota$ denotes the usual contraction of tensor fields, i.e. $(\iota (\d f) \Lambda) (\d g) = \Lambda (\d f, \d g)$ and so on.
\end{definition}

\begin{proposition}
\label{propPoissonBr}
Tensor fields $\Lambda$ and $R$ define respectively a Poisson bracket and a symmetric product of functions as
\begin{equation}
	\{f,g\} = \Lambda (\d f, \d g), \quad
	(f,g) = R (\d f, \d g), \quad
	\forall f,g \in C^\infty (\O^*).
\end{equation}
Furthermore, their restriction to the set $(\O^*)^*$ of $\mathbb{R}$-linear functions defines a Lie-Jordan structure with products
\begin{equation}
	(f_a, f_b) = f_{a\odot b}, \quad \{f_a, f_b\} = f_{\Ll a,b \Lr}, \quad a,b \in \O.
\end{equation}
\end{proposition}
\begin{proof}
The composition laws can be rewritten in terms of Hamiltonian and gradient vector fields as $\{f,g\} = X_f (g)$ and $(f,g) = Y_f (g)$. These relations and the defining properties of the composition laws in Lie-Jordan algebras, presented in Definition \ref{defLJAlg}, prove the asserted statements.
\end{proof}

\begin{remark}
Observe that Hamiltonian vector fields of $\mathbb{R}$-linear
functions are also derivations of the algebras $((O^*)^*, \{\cdot,
\cdot\})$ and $((O^*)^*, (\cdot, \cdot))$ of $\mathbb{R}$-linear
functions. In the language of Dirac \cite{Dirac1964}, Hamiltonian
vector fields are both $c$-derivations and
$q$-derivations. {\color{black} On the other hand, gradient vector
  fields of linear functions are not derivations of the defined brackets.}
\end{remark}

\begin{proposition}
The commutator of two Hamiltonian vector fields is a Hamiltonian vector field. More specifically,
\begin{equation}
	[X_f, X_g] = X_{\{f,g\}}, \quad f,g \in C^\infty (\O^*).
\end{equation}
Moreover, Hamiltonian and gradient vector fields corresponding to $\mathbb{R}$-linear functions satisfy the following commutation relations:
\begin{equation}
	[Y_{f_a}, Y_{f_b}] = -X_{\{f_a, f_b\}} = -X_{f_{\Ll a,b \Lr}}, \quad
	[X_{f_a}, Y_{f_b}] = Y_{\{f_a, f_b\}} = Y_{f_{\Ll a,b \Lr}}, \quad
	a,b \in \O.
\end{equation}
\end{proposition}
\begin{proof}
If $X_f$ and $X_g$ are Hamiltonian vector fields, then, for each $h \in C^\infty (\O^*)$, we have
$[X_f, X_g] (h) = \{f, \{g,h\} \} - \{g, \{f,h\}\} = \{ \{f,g\}, h\} = X_{\{f,g\}} (h),
$ and similarly for the other identities when properly restricted to $\mathbb{R}$-linear functions.
\end{proof}

Hamiltonian vector fields generate the action of the unitary group on $\O^*$ \cite{Grabowski2006}. Thus, when they are considered along with gradient vector fields, they generate an action of the complexification of the unitary group, i.e. the general linear group. It is possible to rephrase these statements in terms of the generalised distributions $D_\Lambda$ and $D_R$ induced respectively by $\Lambda$ and $R$. These are defined as the distributions spanned by Hamiltonian and gradient vector fields. We can also consider the generalised distribution $D_1$ spanned by both types of vector fields.
{\color{black}
\begin{equation}
\begin{gathered}
D_\Lambda (\xi) = {\rm span}_{\mathbb{R}} \left( X_f (\xi) \mid f \in C^\infty(\O^*) \right), \quad
D_R (\xi) = {\rm span}_{\mathbb{R}} \left( Y_f (\xi) \mid f \in C^\infty(\O^*) \right), \\
D_1 (\xi) = {\rm span}_{\mathbb{R}} \left( X_f(\xi), Y_f (\xi) \mid f \in C^\infty(\O^*) \right),
\end{gathered}
\quad \xi \in \O^*.
\end{equation}
}

As proven in \cite{Grabowski2006}, the distributions $D_\Lambda$ and $D_1$ on $\O^*$ are involutive and can be integrated to generalised foliations $\mathcal{F}_\Lambda$ and $\mathcal{F}_1$, respectively. We can characterise their leaves in a simple way if we consider the finite dimensional case. For physical systems, the C*-algebra $\A$ is identified with the algebra of endomorphisms on a finite-dimensional complex Hilbert space $\mathcal{H}$. The algebra $\A$ can be endowed with a Hilbert space structure, with a Hermitian product defined by the trace of operators:
\begin{equation}
\label{prodTr}
\langle a,b \rangle_\A = {\rm tr} (a^* b), \quad
\forall a, b \in A.
\end{equation}
As $\A$ is finite dimensional, the inner product defines a bijection between elements in $\A$ and $\mathbb{C}$-linear functionals in the dual space $\A^*$.

With this identification, the element $a^*$ denotes the Hermitian conjugate of $a \in A$. The set of observables $\O$ is therefore the subset ${\rm Herm} (\mathcal{H})$ of Hermitian operators on $\mathcal{H}$. And due to the canonical isomorphism between $\O$ and $\O^*$ induced by the inner product, we have $O^* \cong \O = {\rm Herm} (\mathcal{H})$. The following property can thus be stated.

\begin{proposition}
\label{ActGL}
{\bf \cite{Grabowski2005}}
With the relation $O^* \cong \O = {\rm Herm} (\mathcal{H})$, the leaves of the generalised foliation $\mathcal{F}_\Lambda$ on $\O^*$ correspond to the orbits of the action of the unitary group $U(\mathcal{H})$ on $\O^*$ defined by $(U,\xi) \mapsto U \xi U^*$. Similarly, the leaves of $\mathcal{F}_1$ correspond to the orbits of the action of the general linear group $GL(\mathcal{H})$ on $\O^*$ defined by $(T,\xi) \mapsto T \xi T^*$.
\end{proposition}

In connection with Quantum Mechanics, Hamiltonian vector fields model the unitary evolution of quantum systems, usually given by the Schr\"odinger or the Heisenberg equation.
{\color{black}
Thus, the foliation $\mathcal{F}_\Lambda$ is related with the
reachable states by unitary evolution from a given one. Evidently,
other types of evolution require additional vector fields. This is the
case of the models by Kaufman \cite{Kaufman1984} and Morrison
\cite{Morrison1986}, which introduce dissipation in the unitary
evolution by means of an entropy function and a symmetric product.
In their case, however, the two tensors -symmetric and skew-symmetric-
do not need to be compatible in the sense of Lie-Jordan
algebras. Therefore their dynamical vector fields could be expressed
in terms of our Hamiltonian and gradient vector fields, but with
coefficients that, in general, will depend on the point (i..e, with
coefficients which are functions).

Other types of
dynamics may require new types of vector fields, in particular when
the vector field is asked to be the generator of a semi-group rather
than a group (although locally). Such is the case of Markovian
dynamics, which will be presented in Section \ref{secKL}. 
}

\subsection{Tensor fields on the set of physical states}
\label{subsecTensF}

In the Heisenberg picture, states of a quantum system are represented by normalised positive functionals on observables, as introduced by F. Strocchi \cite{Strocchi2005}. Thus, they form a subset of the dual space $\O^*$. This section presents a description of the geometrical properties and objects that can be found on the set of quantum states. Neither the $\Lambda$ and $R$ tensor fields nor the linear functions associated to observables can be restricted satisfactorily to the subset of states. For example, it is immediate to check that the normalisation condition is not preserved by gradient vector fields. New geometrical objects representing the physical properties of the system have to be defined. This section describes how these objects can be easily found. In particular, observables will be represented by expectation value functions. The results are however satisfactory, as the new functions are physically relevant, as well as the new found tensor fields representing the Lie and Jordan algebraic properties of observables.

\begin{definition}
\label{defRho}
Let $\O$ be the Lie-Jordan algebra of observables of a quantum system. A state of the system is a $\mathbb{R}$-linear functional $\rho: \O \rightarrow \mathbb{R}$ such that
\begin{equation}
\label{eqDefRho}
	\rho (I) = 1, \quad
	\rho (a^2) \geq 0, \, \forall a \in \O.
\end{equation}
The set of states is denoted by $\D$.
\end{definition}

The set $\D$ is a convex {\color{black} sub}set of $\O^*$. Given two states $\rho_1$, $\rho_2$, when $\lambda_1, \lambda_2 \geq 0$ and $\lambda_1 + \lambda_2 = 1$ the linear combination $\lambda_1 \rho_1 + \lambda_2 \rho_2$ is a state. Points in $\D$ that cannot be written as a non-trivial convex combination of states are called extremal points or pure states. Convex combinations of pure states are called mixed states.
{\color{black}
Observe that $\D$ is not a linear subspace, as the positivity
condition in \eqref{eqDefRho} imposes a constraint on the set, which
is not compatible with linearity. Thus, it is not possible to consider
arbitrary linear combinations of states, only convex ones. 
}

This description of states is equivalent to the usual algebraic description in terms of density operators. Recall that density operators on a Hilbert space $\mathcal{H}$ are positive unit-trace operators. As explained in the previous section, for physical systems $\O^* \cong \O = {\rm Herm} (\mathcal{H})$. Thus, the properties of density operators are equivalent to \eqref{eqDefRho}. Due to this equivalence, it is possible to classify states according to the rank of the corresponding density operators, as done in \cite{Grabowski2005, Grabowski2006}. The following propositions show the stratified structure of the set $\D$ of quantum states.

\begin{proposition}
\label{PAStrat}
Let $\P$ denote the {\color{black} conical} subset of $\O^*$ composed of positive $\mathbb{R}$-linear functionals on $\O$. This subset is a stratified manifold,
\begin{equation}
\P = \bigcup_{k=0}^n \P^k,
\end{equation}
where the stratum $\P^k$ is the set of rank $k$ elements of $\P$. Each stratum $\P^k$ is a leaf of the foliation $\mathcal{F}_1$ corresponding to the distribution $D_1$ generated by Hamiltonian and gradient vector fields.
\end{proposition}

\begin{proposition}
\label{DAStrat}
The set of states $\D$ is a stratified manifold,
\begin{equation}
\D = \bigcup_{k=1}^n \D^k, \quad \mbox{where} \quad
\D^k = \P^k \bigcap \{ \xi \in \O^* | \xi(I) = 1\}.
\end{equation}
\end{proposition}

As stratified manifolds, both $\P$ and $\D$ lack differentiable structures as a whole (although each stratum is itself a smooth manifold, see \cite{Grabowski2005}). However, both sets can be considered as embedded into $\O^*$. They will inherit the differential calculus from $\O^*$. Thus, from now onwards, differential calculus on $\P$ and $\D$ will be implicitly defined with respect to the differentiable structure of $\O^*$.

Observe that, by Proposition \ref{DAStrat}, the set $\D$ is a subset of $\P$ with the constraint $\xi(I) = f_I (\xi) = 1$. The properties of constrained submanifolds were studied by P. A. M. Dirac \cite{Dirac1964}, while considering gauge symmetries of singular Lagrangians, in the context of Hamiltonian Mechanics. In subsequent years, the procedure introduced by Dirac was described in more geometric terms, mostly dealing with the geometry of presymplectic manifolds. It is nevertheless possible to extend such procedure also to symmetric brackets. In the case of quantum states, we need to deal only with one constraint, namely $f_I (\xi) = 1$. Thus, because of the simplicity of the computations, we will avoid technical details and we will merely exhibit the result of the reduction procedure. Reductions of Lie-Jordan algebras were considered in papers by F. Falceto, L. Ferro, A. Ibort and G. Marmo \cite{Falceto2013, Falceto2013a, Falceto2013b}.

Let us consider the foliation of $\O^*$ defined by the gradient vector field $Y_I$. As $Y_I \in D_1$, by Proposition \ref{ActGL} any leaf that intersects $\P$ belongs completely to $\P$. Therefore, we can consider only the leaves of positive functionals. Notice that the functional $0 \in \P$ is a fixed point of $Y_I$. By removing it, we obtain a regular foliation by the orbits of $Y_I$ of $\P_0 := \P - \{0\}$. We can thus define the corresponding quotient manifold by identifying points on the same leaf; two points $\zeta, \zeta'$ are equivalent if $\zeta = c \, \zeta'$, with $c>0$. The set of states $\D$ is a section of this fibration defined by the elements of trace equal to one.

We are interested in the characterization of geometrical objects in
$\D$ as objects in $\P$ {\color{black} (i.e., as normalized operators)}
that are projectable with respect to the 
fibration associated with dilations. {\color{black} In this respect,
  states are in one-to-one correspondence with fibers -or equivalence
  classes- of this fibration.} This procedure is analogous to the geometrical description of the projective Hilbert space in the Schr\"odinger picture, sketched in the Introduction and presented in \cite{Clemente-Gallardo2007, Clemente-Gallardo2008, Clemente-Gallardo2009, Clemente-Gallardo2012}. Therefore, let us consider the projection $\pi_\P : \P_0 \rightarrow \D$ defined as
\begin{equation}
\label{defpiP}
\pi_\P (\zeta) = \frac{1}{f_I (\zeta)}\, \zeta, \quad \zeta \in \P_0.
\end{equation}
{\color{black}

  \begin{remark}
    Notice that the space of density states may be thought of in two
    different ways: on one hand as 
    a quotient space of the conical set  $\mathcal{P}$ under
    dilations; and, on the other hand, as the submanifold of
    $\mathcal{P}$ defined by the intersection with the affine subspace
    $$
f_{I}(\xi)=1.
$$
This double aspect of $\mathcal{D}$ allows to consider the pullback of
covariant tensor fields when it is thought of as a submanifold, and to
consider the projection of contravariant tensor fields when it is
thought of as a quotient manifold.  We shall use the same notation for
the quotient space and for the section, where a representative is
selected by the normalization condition. 
  \end{remark}

Instead of evaluation functions (which are linear functions on the
dual space), new objects are required in order to properly represent
observables on the set of states $\D$. 

\begin{definition}
\label{dfnEVF}
For any observable $a \in \O$, the expectation value function $e_a: \O^* -\{0\} \to \mathbb{R}$ is the smooth function defined as
\begin{equation}
e_a (\xi) := \frac{f_a (\xi)}{f_I(\xi)}, \quad
\xi \in \O^*-\{0\}, \quad a \in \O.
\end{equation}
\end{definition}

\begin{proposition}
The set $\E$ of expectation value functions is a linear
space. Gradient and Hamiltonian vector fields associated by Definition
\ref{dfnXfYf} to expectation value functions are projectable with
respect to the projection \ref{defpiP}. As these vector fields
preserve the normalization, they are also tangent to the set of states
as a submanifold of the positive operators.

\end{proposition}
\begin{proof}
Linearity of the space $\E$ is clear from Definition \ref{dfnEVF}. For any $e_a \in \E$, let us evaluate the gradient vector field $Y_{e_a}$ on the function $f_I$ that determines the set $\D$.
\begin{equation}
Y_{e_a} (f_I) = (e_a, f_I) = \left( \frac{f_a}{f_I}, f_I \right) = \frac{1}{f_I} (f_a, f_I) - \frac{f_a}{f_I^2} (f_I, f_I) = \frac{1}{f_I} f_a - \frac{f_a}{f_I^2} f_I = 0.
\end{equation}
Notice that $(f_a,f_I)=f_{a\odot I}=f_a$. 
It is even easier to prove that $\{e_a, f_I\} = 0$. Therefore, Hamiltonian and gradient vector fields of expectation value functions preserve the set of states $\D$.
\end{proof}

Observe that expectation value functions and evaluation functions are
related in the following manner:
\begin{equation}
e_a (\zeta) = \pi_\P^* (f_a|_{\D}) (\zeta) = \frac{f_a (\zeta)}{f_I(\zeta)}, \quad
\zeta \in \P_0, \quad a \in \O,
\end{equation}
Thus, expectation value functions are projectable onto $\D$. They can
also be thought of as the extension, to the whole conical space of
positive operators, of the evaluation functions restricted to the
normalized section. The extension is being made by declaring the extension
to have constant value on each equivalence class. 
}


The Poisson brackets and symmetric products of expectation value functions read
\begin{equation}
\label{prodExpVF}
\begin{aligned}
\{e_a, e_b\} (\zeta) & = \Lambda (\d e_a, \d e_b) (\zeta) = \dfrac{e_{\Ll a,b \Lr} (\zeta)}{f_I(\zeta)}, \\[0.3cm]
(e_a, e_b) (\zeta) & = R (\d e_a, \d e_b) (\zeta) = \dfrac{e_{a \odot b} (\zeta)}{f_I(\zeta)} - \dfrac{e_{a} (\zeta) e_{b} (\zeta)}{f_I(\zeta)},
\end{aligned} \quad a,b \in \O, \quad \zeta \in \P_0,
\end{equation}
and as a particular case of the last relation 
$$(e_a, e_a) =\frac{e_{a^2}-(e_a)^2}{f_I}.
$$
{\color{black} As a consequence, the linear space $\E$ is not a
  Lie-Jordan algebra with respect to these composition laws
  because the product of two expectation-value functions is not an
  expectation value function.
}

We are interested in reproducing on $\D$ the algebraic properties of
the Lie-Jordan algebra of observables. In other words, we need
bivector fields that, when acting on expectations value functions,
give results in terms only of expectation value functions. The
formulae above suggest that we have to define new bivector fields
$\widehat{\Lambda}$ and $\widehat{R}$ whose values at points $\zeta
\in \P_0$ are {\color{black} conformally related to the original ones}
\begin{equation}
\widehat{\Lambda} (\zeta) = f_I(\zeta) \Lambda(\zeta), \quad
\widehat{R} (\zeta) = f_I(\zeta) R(\zeta)
\end{equation}
These tensor fields are by construction projectable under
dilation. Their evaluation on expectation value function at points
$\rho \in \D$ give tensor fields on the set of states. With some abuse
of notation, let us denote by $e_a$ also their restriction to the
manifold $\D$. The set of expectation value functions on $\D$ will be
denoted as $\E_\D$. 

\begin{theorem}
\label{thTF}
The restrictions of $\widehat{\Lambda}$ and $\widehat{R}$ to $\D$ are a pair of tensor fields $\Lambda_\D$ and $R_\D$ on $\D$ whose actions on expectation value functions are given by
\begin{equation}
\label{LambdaFieldD}	
\Lambda_\D (\d e_a, \d e_b) (\rho) = e_{\Ll a, b \Lr} (\rho), \quad
R_\D (\d e_a, \d e_b) (\rho) = e_{a\odot b} (\rho) - e_a (\rho) e_b (\rho),
\end{equation}
for any $\rho \in \D$ and any $a,b \in \O$.
\end{theorem}
\begin{proof}
These expressions are simply restrictions of \eqref{prodExpVF}. Consider the constant function $e_I =1$. As $\d e_I = 0$, the action of any tensor field on it should be zero. Indeed,
\begin{equation*}
\Lambda_\D (\d e_a, \d e_I) (\rho) = e_{\Ll a, I \Lr} (\rho) = 0, \quad
R_\D (\d e_a, \d e_I) (\rho) = e_{a\odot I} (\rho) - e_a (\rho) e_I (\rho) = e_a (\rho) - e_a (\rho) = 0,
\end{equation*}
for any $\rho \in \D$ and any $a\in \O$, as expected.
\end{proof}

\begin{theorem}
\label{propLJE}
The Lie-Jordan algebra of observables can be recovered as the algebra of expectation value functions $\E_\D$ on $\D$, with respect to the following Poisson and Jordan brackets defined on functions $e_a$ on $\D$.
\begin{equation}
\begin{aligned}
\label{PoissonJordanBrackets}
\{e_a, e_b\}_\D (\rho) & = \Lambda_\D (\d e_a, \d e_b) (\rho) = e_{\Ll a,b \Lr} (\rho), \\
(e_a, e_b)_\D (\rho) & = R_\D (\d e_a, \d e_b) (\rho) + e_a (\rho) e_b (\rho) = e_{a\odot b} (\rho),
\end{aligned}
\end{equation}
for any $a,b \in \O$ and any $\rho \in \D$.
\end{theorem}

{\color{black}
\begin{corollary}
  Consider the complexification of the vector space  of expectation value
  functions $\E_\D$, i.e., the extension to complex combinations of
  real elements of the form
  $$
e_{a}+i e_{b}, \qquad e_{a}, e_{b}\in \mathcal{E}_{\mathcal{D}},
$$
considered as a complex vector space.
We can consider  an associative complex algebra by defining the following product:
\begin{equation}
e_a * e_b = (e_a, e_b)_\D +i \{e_a, e_b\}_\D = e_{ab}, \quad a,b \in \O,
\end{equation}
where $ab$ is considered to be defined by Equation \ref{eq:7}.

\end{corollary}

  Clearly,
  $$
(e_a, e_b)_\D=\frac 12 \left ( e_{a}*e_{b}+e_{b}*e_{a} \right) ,
\qquad
\{e_a, e_b\}_\D=\frac {-i}2 \left ( e_{a}*e_{b}-e_{b}*e_{a} \right).
$$

Notice that this associative operation is clearly non-local, since it
is encoding the associative algebra of observables.  Thus, we can
claim that the tensor $R_{\mathcal{D}}$, which encodes the difference
between the (symmetrized) associative product and the pointwise
(local) products of functions represents the ``deviation from locality
'' of the $*$--product. Here locality means that the support of the
product is, in general,  not included in the intersection of the
supports of the factors.

\begin{definition}
Let $\widetilde{X}_g$ and $\widetilde{Y}_g$ denote, respectively, the Hamiltonian and gradient vector fields on $\D$, that is, the evaluations of $R_\D$ and $\Lambda_\D$ on the exact 1-form $\d g$, i.e.
\begin{equation}
\label{defGradHam}
	\widetilde{X}_g = \iota (\d g) \Lambda_\D, \quad
	\widetilde{Y}_g = \iota (\d g) R_\D.
\end{equation}
for any $g \in C^\infty (\D)$. In particular, for expectation value functions, we use the simplified notation $\widetilde{X}_{e_a} = \widetilde{X}_a$, $\widetilde{Y}_{e_a} = \widetilde{Y}_a$, with $a \in \O$.
\end{definition}

\begin{proposition}
The commutators of Hamiltonian and gradient vector fields are
\begin{equation}
	[\widetilde{X}_a, \widetilde{X}_b] = \widetilde{X}_{\Ll a,b \Lr}, \quad
	[\widetilde{Y}_a, \widetilde{Y}_b] = -\widetilde{X}_{\Ll a,b \Lr}, \quad
	[\widetilde{X}_a, \widetilde{Y}_b] = \widetilde{Y}_{\Ll a,b \Lr},
\end{equation}
for any $a,b \in \O$.
\end{proposition}
\begin{proof}
	The result follows by use of Jacobi identity, Jordan identity and relations \eqref{LieJordanRel}.
\end{proof}

\begin{corollary}
For a quantum $n$-level system, the Hamiltonian and gradient vector fields of expectation value functions span the complexification $\mathfrak{sl} (n, \mathbb{C})$ of the Lie algebra of the special unitary group $\mathfrak{su} (n)$.
\end{corollary}

We have identified within our geometric formalims a natural
description of Hamiltonian and gradient vector fields associated with
the canonical tensors which encode the algebraic properties of the
$C^{*}$--algebra of observables of a quantum system, thought as
functions on the space of density states.  These two sets
of vector fields generate distributions ($D_{\Lambda}$ and $D_{1}$)
which encode the natural coadjoint orbits of the unitary group on
$\mathcal{D}$ and the stratified structure of $\mathcal{D}$
corresponding to a nonlinear action of the general linear group (see
\cite{Grabowski2005} ).  Furthermore, it is also possible to define,
on each stratum $\mathcal{D}^{k}$ of $\mathcal{D}$,  a (1,1)
tensor satisfying
\begin{equation}
  \label{eq:1}
  J_{\mathcal{D}^{k}}(\rho) (\widetilde X_{f}(\rho)):=\widetilde Y_{f}(\rho), \qquad f\in
  \mathcal{E}_{\mathcal{D}}, \quad \rho\in \mathcal{D}^{k}.
\end{equation}
Notice that the rank of the tensor depends on the point. Thus, on each
orbit of the unitary group it would define a complex structure which
may have different rank depending on the orbit. In any case, we can 
always write that
\begin{equation}
  \label{eq:2}
  J_{\mathcal{D}^{k}}^{3}=-J_{\mathcal{D}^{k}}, \qquad \forall k.
\end{equation}
}

\begin{remark}
In connection with the probabilistic nature of Quantum Mechanics, the tensor field $R_\D$ is related with the definitions of variance ${\rm Var} (a)$ and covariance ${\rm Cov} (a,b)$ of observables:
\begin{align*}
& {\rm Var} (a) (\rho) = R_\D (\d e_a, \d e_a) (\rho) = e_{a^2} (\rho) - (e_a(\rho) )^2, \\
& {\rm Cov} (a,b) (\rho) = R_\D (\d e_a, \d e_b) (\rho) = e_{a \odot b} (\rho) - e_a(\rho) e_b(\rho) ,
\end{align*}
i.e. the variance and covariance in terms of expectation values of
observables. The relation between Jordan algebras and statistics was
already present in the original works by P. Jordan
\cite{Jordan1933,Jordan1934}. Future works will further develop the
importance of this tensor field; for now, it is enough to consider
that it represents the Jordan product of observables and is therefore
necessary to properly describe their algebraic properties.

\color{black}

It is important to remark that the notions of variance and covariance
can also be considered in   classical setting, where we consider
observables that pairwise commute.
Nonetheless, there is an important difference between the
classical and the quantum definitions: classical  variance and covariance
  always vanish for pure states (considered to be the extremals of the
  simplex of ``classical'' states). But quantum variance and covariance
  are generically non-vanishing on quantum pure states, as it happens
  with  the relations  encoded in the
  $R_{\mathcal{D}}$ tensor. Indeed, if we compute the tensor
  for a generic pure state $\rho\in \mathcal{D}^{1}$, the result does not vanish for any
  generic operator $A$, unless $\rho$ happens to be an eigenstate of
  $A$.  From this point of view  we can claim that the tensor 
  is capturing a genuine quantum   feature. 
\end{remark}
}

{\color{black}

\begin{remark}
Observe that, in this geometric setting, expectation value functions
are the objects that represent the observables  as functions on the
space of quantum states. This description 
is in direct connection with the Ehrenfest 
theorem, which, in its usual formulation, describes the evolution of
expectation values of  a pair of canonically conjugated observables,
as the position and the linear momentum.  In our case, as we are
considering  finite dimensional quantum systems, the action of the
Hamiltonian vector field  associated to the function $e_{H}\in
\mathcal{D}$ on the set  $\mathcal{E}_{\D}$ is providing a
finite dimensional geometric analog of the Ehrenfest theorem.
Furthermore, the geometric
formalism here presented is in a sense a generalisation of the
Ehrenfest approach, as it applies to both pure and mixed states,
besides offering a more intrinsic formulation in terms of differential
geometry. Ehrenfest theorem (on pure states) is recovered by Theorem
\ref{propLJE} when restricted to the stratum $\D_1$. 

From a more general perspective and now in  the case of infinite
dimensional quantum systems, we can also consider the recent
paper by Bonet-Luz, Ohsawa and Tronci \cite{Bonet-Luz2015,
  Bonet-Luz2016, Ohsawa2016}. In these works the evolution of the
average value of physical magnitudes is considered in connection with
an action of the \textit{Ehrenfest group} (the product of the
{\color{black} Heisenberg}
 and the unitary groups), and Ehrenfest theorem is proved to
 correspond to the flow of a  Lie-Poisson  Hamiltonian system
 associated to the  corresponding Lie  algebra.  Furthermore, the
 expressions of the variance and the covariance of the quantum system
which we have encoded in the tensor $R_{\mathcal{D}}$ is obtained in
that context as  a momentum map associated to the Ehrenfest group. 
A detailed comparison of our approach with  the quoted one, where
expectation value functions may be thought of as joint functions of
the observables and of the states, will be taken up in a future work. 
\end{remark}

Summarizing, we have obtained a geometric characterization of the set
of states $\D$ along with a realization of the C*-algebra by means of
expectation value functions defined on it .This has been achieved by
introducing  compatible  Lie products and Jordan products ,i.e.,a
Lie-Jordan product. As it is  a manifold with a non-smooth boundary, the 
differentiable structure of the set $\D$ is described in terms of a
larger differentiable manifold of which $\D$ is a subset. The
Lie-Jordan structure gives raise to a pair of compatible tensor
fields  $\Lambda_\D$ and $R_\D$ (respectivelty skewsymmetric and symmetric)
that reproduce the algebraic properties of observables.  These tensor
fields have additionally physical relevance. The Poisson tensor field
$\Lambda_\D$ characterises the unitary evolution of quantum systems,
represented by the von Neumann equation in the language of density
operators. On the other hand, the symmetric tensor field $R_\D$ is
related to the variance and covariance of observables. 
}
{\color{black}
\subsection{Hamiltonian and gradient vector fields: 
  Lie-Jordan algebras and dissipation}

When a physical system interacts with the environment, loss of
energy, creation of entropy or decoherence for quantum systems cannot be
ignored. 
The symplectic or Poisson descriptions of the dynamics do not  apply
directly  to dissipative systems because the Hamiltonian usually has
the meaning of energy and it  would be conserved.
To describe  dissipative systems also in terms of functions (say the
free energy, entropy or energy in thermodynamics) it has been proposed
to use a combination of symmetric and skew-symmetric contravariant
tensors.  This description in terms of a symplectic and a metric
tensor has been  called a metriplectic description and  was introduced
by  Kaufman \cite{Kaufman1984} and Morrison 
\cite{Morrison1986}, mostly having in mind thermodynamical systems
but also other fluids and plasma physics.  
As in our picture we do have a skewsymmetric tensor to realize
commutator brackets among expectation value functions  and a symmetric
tensor to realize the Jordan product, it is quite natural to compare
our situation with the one of metriplectic structures do describe
dissipation. 

For instance, we may consider
modification of the Heisenberg equation 
by an additional term that incorporates dissipation: 
\begin{equation}
\label{eqKM}
\frac{\d A}{\d t} = -\Ll H,A\Lr + S \odot A, \quad A, H, S \in \O,
\end{equation}
with $H$ the Hamiltonian operator of the system and $S$ an observable
associated with the entropy of the system. Observe that this model is
related to the one considered by  Rajeev \cite{Rajeev2007}, where,
however, the symmetric and the skewsymmetric product are related by a
(1,1) tensor field which represents the complex structure. See also
the work by Benvegnu, Sansonetto 
and Spera \cite{Benvegnu2004} for the description of other dissipative
systems. The representation of equation \eqref{eqKM} on the set of
quantum states $\D$ determines the following vector field: 
\begin{equation}
Z_{KM} = \widetilde{X}_{H} + \widetilde{Y}_S,
\end{equation}

As in our case the Hamiltonian  vector field may be related to the
gradient vector field by means of a (1-1)-tensor field , it is possible
to extract further consequences on the behaviour of the dynamics
described by combinations of them. Let us restrict our considerations
to pure states, for simplicity. 
Consider a Hamiltonian vector field $\widetilde{X}_a$ and a gradient
vector field $\widetilde{Y}_a$ associated to the same observable $a
\in \O$. The expectation value function $e_a$ is enough to determine
the eigenvectors and eigenvalues of $a$ \cite{Clemente-Gallardo2008,
  Clemente-Gallardo2009}. Indeed,  critical points of the expectation value
function correspond to the eigenvectors while the values at such
points correspond to the eigenvalues.
The critical points for the expectation value functions  will be
equilibrium points for both vector fields, the Hamiltonian and the
gradient. This circumstance allows to make considerations on the
stability of these equilibria with respect to the total vector field. 
We shall not indulge further on these aspects because they go beyond
linear dynamics on the space of quantum states. 


More complex models 
\cite{Bloch1996, Brody2008a, Brody2009, Gisin1981}  can no longer
be described simply by linear combinations of Hamiltonian and gradient
vector fields. It is thus necessary to consider more general
expressions of vector fields on $\D$. For instance, in the case of
Gisin model \cite{Gisin1981}, aiming at the description of nonlinearities in
the quantum dynamics, the evolution of the system is written
as a double bracket of the form
\begin{equation}
  \label{eq:3}
  \dot \rho=\Ll\rho, \Ll\rho, H\Lr\Lr.
\end{equation}
It is straightforward to prove that this dynamics, even if
nonlinear, preserves the purity of the quantum state
$\rho$. Indeed, the dynamics is isospectral:
\begin{equation}
  \label{eq:4}
 \frac d{dt} \mathrm{Tr} \rho^{2}=\mathrm{Tr} (\rho \dot \rho +\dot
 \rho \rho)= 2 \mathrm{Tr}  (  \Ll\rho, \Ll\rho, H\Lr\Lr)=0
\end{equation}
Hence the evolution will take place onto an orbit of the unitary
group.  Thus, it is possible to write the  dynamical vector field as a
nonlinear combination of Hamiltonian  vector fields.

In \cite{Brody2008a, Brody2009} another dynamical system was
introduced by using a double bracket, but in this case in a linear
form:
\begin{equation}
  \label{eq:5}
  \dot G=\Ll H, \Ll H, G\Lr \Lr.
\end{equation}

By using the Lie-Jordan compatibility condition \eqref{LieJordanRel},
it is immediate to prove that the RHS of Eq. \eqref{eq:5} can be
written as a combination of Jordan products:
\begin{equation}
  \label{eq:6}
  \Ll H, \Ll H, G\Lr \Lr=(H\odot G)\odot H-H^{2}\odot G
\end{equation}
Last term corresponds to the gradient vector field $\widetilde
Y_{H^{2}}$, and thus the dynamical system encodes a dissipative
process.
However{\color{black}, as we are considering that the system is finite
  dimensional, we know that} the trace of $\dot G$ is zero, and  therefore the dynamics,
which is obviously linear, 
must preserve the trace. Therefore it cannot be represented by means of linear
combination of gradient vector fields 

The first factor  is indeed  a different  type of vector field
which will be studied in  Section \ref{secKL}, in the context of
Markovian dynamics.  In particular, the example in  Section
\ref{sec:open-3-levels} corresponds precisely to a model of the form
of Equation \eqref{eq:5}.  Notice that the main difference with respect to 
the previous case is the linearity: that precise combination on the
RHS makes the resulting vector field linear. Also, the first term in
the RHS induces changes in the purity of states. Thus, this model
presents two main differences with the Gisin models: linearity and the
non-preservation of purity. In other words, this equation models a
particular type of Markovian evolution of quantum systems. Section
\ref{secKL} presents a particular instance of this model in the case
of 3-level systems.

}
\section{Example: two-level system}
\label{section2levels}

As an explicit example, we consider the space of states of a two-level system. Our aim is to find the expressions of the tensor fields $\Lambda_\D$ and $R_\D$, which allows us to determine the Hamiltonian and gradient vector fields. These vector fields can be plotted, giving some insight into the geometrical properties of the space of states.

The C*-algebra associated to a two-level system is isomorphic to ${\rm End}(\mathbb{C}^2)$. Therefore, both the set of observables $\O$ and its dual space $\O^*$ are isomorphic to the set ${\rm Herm} (2)$ of $2 \times 2$ Hermitian matrices. A basis $\{\sigma_\mu\}_{\mu=0}^3$ of $\O$ is given by the three Pauli matrices and the identity matrix:
\begin{equation}
	\sigma_0 = \begin{pmatrix} 1 & 0 \\ 0 & 1 \end{pmatrix}, \quad
	\sigma_1 = \begin{pmatrix} 0 & 1 \\ 1 & 0 \end{pmatrix}, \quad
	\sigma_2 = \begin{pmatrix} 0 & -i \\ i & 0 \end{pmatrix}, \quad
	\sigma_3 = \begin{pmatrix} 1 & 0 \\ 0 & -1 \end{pmatrix}.
\end{equation}
The Lie and Jordan product of the elements in the basis $\{\sigma_\mu\}$ are\footnote{In the following, indexes denoted by Greek letters will run from 0 to 3, while those represented by Latin letters will take values 1, 2 and 3.}
\begin{equation}
\begin{aligned}
	& \Ll \sigma_j, \sigma_k \Lr = \epsilon_{jkl} \sigma_l, \quad
	\sigma_j \odot \sigma_k = \delta_{jk} \sigma_0, \quad
	& j,k = 1,2,3
	\\
	& \Ll \sigma_0, \sigma_\mu \Lr = \Ll \sigma_\mu, \sigma_0 \Lr = 0, \quad
	\sigma_0 \odot \sigma_\mu = \sigma_\mu \odot \sigma_0 = \sigma_\mu, \quad
	& \mu = 0,1,2,3.
\end{aligned}
\end{equation}

Consider the inner product in $\O$ defined by the trace, $\langle a,b \rangle_A = {\rm tr} (ab)$, as in \eqref{prodTr}, to identify $\O$ and $O^*$. The dual basis in $\O^*$ is therefore $\{\sigma^\mu = \frac 12 \sigma_\mu\}_{\mu=0}^3$. An element $\xi \in \O$ takes the form $\xi = x_\mu \sigma^\mu = \frac 12 x_\mu \sigma_\mu$. In particular, states $\rho \in \D \subset \O^*$ must be unit-trace positive elements, which gives the following result.

\begin{proposition}
The coordinate expression of a state of the two-level system is
\begin{equation}
\label{rho2LCoord}
\rho = \sigma^0 + x_j \sigma^j=
\frac 12 \begin{pmatrix} 1+x_3 & x_1-ix_2 \\ x_1+ix_2 & 1-x_3 \end{pmatrix}, \quad
x_1^2 +x_2^2 +x_3^2 \leq 1.
\end{equation}
\end{proposition}
\begin{proof}
Consider a generic element $\xi = x_\mu \sigma^\mu \in \O^*$. Then:
\begin{equation*}
{\rm tr} (\xi) = x_0, \quad
{\rm tr} (\xi^2) = \frac 12 (x_0^2 + x_1^2 + x_2^2 + x_3^2).
\end{equation*}
An element in $\O^*$ is a state if it satisfies conditions in \eqref{eqDefRho}, or in matrix notation, if it is normalised and positive. This leads to the expression presented in the Proposition.
\end{proof}

Therefore, the set of states is three-dimensional, as it can be parametrised by points $(x_1, x_2, x_3) \in \mathbb{R}^3$ such that their norm is not greater than 1. That is to say, the set of states is parametrised by the solid ball of radius 1 in $\mathbb{R}^3$. This representation of the set of states of the two-level system is called the Bloch ball. Points on the surface of the ball (that is, vectors with radius 1) parametrise states that are rank-1 projectors, as $\rho^2 = \rho$; these are the pure states of the system. The interior of the ball parametrises mixed states. In the language of Proposition \ref{DAStrat}, the set of states is stratified as
\begin{equation}
\label{strat2lev}
\D = \D_1 \cup \D_2,
\end{equation}
with $\D_1$ the surface of the Bloch ball and $\D_2$ its interior.

Observables are represented by expectation value functions on $\D$. With the given basis, the association between observables and expectation value functions is
\begin{equation}
\label{expVF2lev}
a = a^\mu \sigma_\mu \Rightarrow e_a = a^0 + a^j x_j, \quad a^0, a^1, a^2, a^3 \in \mathbb{R},
\end{equation}

\begin{proposition}
The coordinate expressions for the contravariant tensor fields $\Lambda_\D$ and $R_\D$ are
\begin{equation}
\label{LambdaR}
\begin{aligned}
& \Lambda_\D = \frac 12 \epsilon_{jkl} \, x_l \frac{\partial}{\partial x_j} \wedge \frac{\partial}{\partial x_k}, \quad
& R_\D = \frac{\partial}{\partial x_j} \otimes \frac{\partial}{\partial x_j} - x_j x_k \frac{\partial}{\partial x_j} \otimes \frac{\partial}{\partial x_k}.
\end{aligned}
\end{equation}
\end{proposition}
\begin{proof}
Coordinate functions on $\D$ are given by the expectation value functions associated to the Pauli matrices. The values of $\Lambda_\D$ and $R_\D$ on these coordinate functions are the following
\begin{equation*}
\Lambda_\D (\d x_j, \d x_k) (\rho) = e_{\Ll \sigma_j, \sigma_k \Lr} (\rho) = \epsilon_{jkl} x_l, \quad
R_\D (\d x_j, \d x_k) (\rho) = e_{\sigma_j \odot \sigma_k} (\rho) - e_{j} (\rho) e_{k} (\rho) = \delta_{jk} - x_j x_k.
\end{equation*}
From these results follow the coordinate expressions presented in the Proposition.
\end{proof}

As stated in Proposition \ref{propLJE}, the algebra of observables is recovered on $\D$ as the algebra of expectation value functions $\E_\D$ with products $\{e_a, e_b\}_\D = \Lambda_\D(\d e_a, \d e_b)$ and $(e_a, e_b)_\D = R_\D(\d e_a, \d e_b) + e_a e_b$. With the decomposition given in \eqref{expVF2lev}, their explicit expressions are
\begin{equation}
\label{algEVF2Lev}
\{e_a, e_b\}_\D = \epsilon_{jkl} \ a^j b^k x_l = e_{\Ll a,b \Lr}, \qquad
(e_a, e_b)_\D = a^\mu b^\mu + (a^j b^0 + a^0 b^j) x_j = e_{a \odot b}, \quad a,b, \in \O.
\end{equation}
As expected, the algebra $\E_\D$ is isomorphic to the algebra of observables.

{\color{black}
From these expressions, the Hamiltonian and gradient vector fields, defined in \eqref{defGradHam} can easily be computed. As a practical example, consider the Hamiltonian observable associated to a magnetic field ${\bf B} = (B^1, B^2, B^3)$, with expectation value function given by \eqref{expVF2lev}:
\begin{equation}
H = B^j \sigma_j \Rightarrow e_H (\rho) = B^j x_j = {\bf B} \cdot {\bf x}, \quad \rho \in \D,
\end{equation}
with ${\bf x} = (x_1, x_2, x_3)$ the coordinates of $\rho$ by \eqref{rho2LCoord}. The value of function $e_H$ at each state $\rho$ is precisely the energy for that state. It is thus immediate to prove that the extreme values of the energy are obtained at states with coordinates ${\bf x} = \pm \|{\bf B}\|^{-1} {\bf B}$. These are precisely the eigenstates of the Hamiltonian $H$.

Hamiltonian and gradient vector fields for the observable $H$ are
\begin{equation}
\widetilde{X}_H = \epsilon_{jkl} \, x_j B^k \pdx{l}, \quad
\widetilde{Y}_H = B^j \pdx{j} - ({\bf B} \cdot {\bf x}) x_j \pdx{j}.
\end{equation}
Stationary states for $\widetilde{X}_H$  are those states with
coordinate vector ${\bf x}$ parallel to the magnetic field ${\bf
  B}$. In particular, the only stationary pure states are the
eigenstates of $H${\color{black}, which correspond to the critical points of the
  function $e_{H}$ }. Integral curves of $\widetilde{X}_H$ are circles
around the axis parallel to ${\bf B}$. Therefore, every stationary
state is pseudo-stable. Evidently, energy is preserved along the
evolution, as $\widetilde{X}_H (e_H) = 0$. See Figure \ref{figHamB}
for the case ${\bf B} = (0,0,1)$. 

The evolution governed by the gradient vector field is more complex. Stationary states, with coordinates ${\bf x}_S$ satisfy the condition:
\begin{equation}
{\bf B} = ({\bf B} \cdot {\bf x}_S) {\bf x}_S.
\end{equation}
Solutions to this equation can be found by taking the norms of both sides of the equation. Thus,
\begin{equation*}
\|{\bf B}\| = \|{\bf B}\| \|{\bf x}_S\|^2 |\cos \theta| \Rightarrow \|{\bf x}_S\|^2 |\cos \theta| = 1,
\end{equation*}
with $\theta$ the angle between vectors ${\bf B}$ and ${\bf x}_S$. As coordinate vectors are restricted to the unit ball in $\mathbb{R}^3$, the only solutions for the equation satisfy $\|{\bf x}_S\|= 1$ and $\cos \theta= \pm 1$. Therefore,
\begin{equation}
{\bf x}_S = \pm \|{\bf B}\|^{-1} {\bf B},
\end{equation}
which are again the eigenstates of the Hamiltonian $H$.
\color{black}
Regarding stability, consider the change in the energy along the
integral lines of $\widetilde{Y}_{-H}=-\widetilde Y_{H} $, given by
\begin{equation}
\widetilde{Y}_{-H} (e_H) = -\|{\bf B}\|^2 + ({\bf B} \cdot {\bf x})^2 = \|{\bf B}\|^2 (-1+\|{\bf x}\|^2 \cos^2 \theta) \leq 0,
\end{equation}
equality holding only for the stationary states. As a consequence, the
ground state of the system is stable under this evolution, while the
excited state is unstable. Hence, this gradient vector field is
describing a dissipative physical process.  The opposite behaviour (i.e., evolution from
an (stable) excited stated into a (unstable) ground state) can be obtained by
reversing the sign of the vector field, i.e., $\tilde Y_{H}$. In this
case, we would be modelling a system which is receiving energy from
the environment. If we start from any state different from an
equilibrium state, the flow will reach the excited state.
\color{black}

An important characteristic of the evolution due to the gradient vector field $\widetilde{Y}_H$ is the preservation of pure states. Figure \ref{figGradB} shows the gradient vector field for ${\bf B} = (0,0,1)$. Observe that, on the surface of the Bloch ball, which corresponds to the stratum of pure states, the values of the vector field are always tangent to the surface. This means that no mixing of pure states occurs due to the evolution governed by $\widetilde{Y}_H$.

\begin{figure}[t]
\centering
\begin{subfigure}[ht]{.4\textwidth}
\centering
\includegraphics[width=\textwidth]{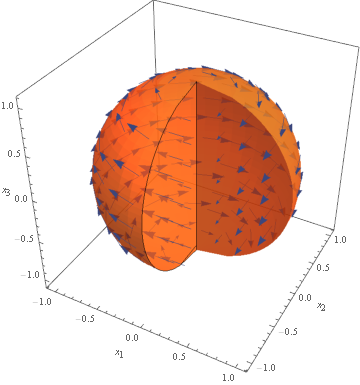}
\caption{Hamiltonian vector field $\widetilde{X}_H$.}
\label{figHamB}
\end{subfigure}
\quad
\begin{subfigure}[ht]{.4\textwidth}
\centering
\includegraphics[width=\textwidth]{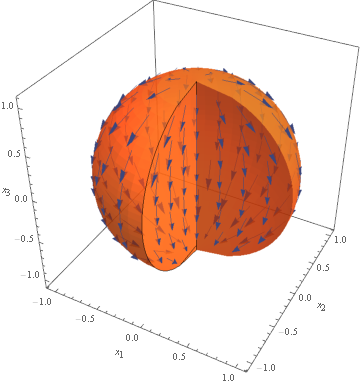}
\caption{Gradient vector field $-\widetilde{Y}_H$.}
\label{figGradB}
\end{subfigure}
\caption{\color{black} Hamiltonian and gradient vector fields for a
  magnetic field ${\bf B} = (0,0,1)$. Ground and first excited states
  for the Hamiltonian $H = B^j \sigma_j= \sigma_3$ correspond
  respectively to the  south  and north poles of the Bloch ball,
  respectively. These are stationary states for both vector
  fields. Additionally, all states along the $x_3$-axis are also
  stationary states for $\widetilde{X}_H$, as their coordinate vectors
  are parallel to ${\bf B}$. Integral curves for $\widetilde{X}_H$ are
  circumferences around this axis. For the gradient vector field
  $-\widetilde{Y}_H$, integral curves go from the excited  state to the
  ground state. In both cases, the stratum of pure states,
  i.e. the surface of the Bloch ball, is preserved along evolution.} 
\end{figure}
}

The Hamiltonian and gradient vector fields associated to the coordinate functions $x_j$, respectively $\widetilde{X}_j$ and $\widetilde{Y}_j$, are
\begin{equation}
\label{vfxj}
\widetilde{X}_j = \epsilon_{jkl}\, x_l \pdx{k}, \quad
\widetilde{Y}_j = \pdx{j} - x_j x_k \pdx{k}, \quad j=1,2,3.	
\end{equation}
whose commutators are
\begin{equation}
[\widetilde{X}_j, \widetilde{X}_k] = \epsilon_{jkl}\, \widetilde{X}_l, \quad
[\widetilde{Y}_j, \widetilde{Y}_k] = -\epsilon_{jkl}\, \widetilde{X}_l, \quad
[\widetilde{X}_j, \widetilde{Y}_k] = \epsilon_{jkl}\, \widetilde{Y}_l, \quad
j,k = 1,2,3.
\end{equation}
The vector fields close on the complexification $\mathfrak{sl} (2, \mathbb{C})$ of the Lie algebra of the special unitary group $\mathfrak{su} (2)$.

\section{Geometric description of the Kossakowski-Lindblad equation}
\label{secKL}

{\color{black}
Section \ref{subsecTensF} shows that the set $\D$ of states is a
manifold with boundary, thus differentiable calculus is possible by
considering its embedding into $\O^*$. Previous sections show the
geometric characterisation of algebraic structures in terms of tensor
fields, and the distributions generated by them. Now, we present a
more general set of vector fields on $\D$, those describing the
Markovian evolution of open quantum systems.  As we are going to see,
we will have to consider, besides Hamiltonian and Gradient vector
fields, a new class of fields in order to reproduce Markovian evolution.
}

A system
that is not isolated, but interacting with a certain environment, is
called open. Let us consider a quantum $n$-level system coupled with
some external environment, such as an electromagnetic field or a
thermal bath. The evolution of the open quantum system is said to be
Markovian if it depends only on the present state of the system and is
independent of the states at previous times. Hence, it is usually said that
the system `has no memory'. The study of Markovian evolution of open
quantum systems was given a formal description by Gorini, Kossakowski
and Sudarshan \cite{Gorini1976a} and by Lindblad \cite{Lindblad1976}. 

\begin{theorem}[\bf The Kossakowski-Lindblad operator]
Let $\mathcal{H}$ be a $n$-di\-men\-sio\-nal Hilbert space describing a quantum system. Let $\rho_0 \in \D$ be the initial state of the system, with $\A = {\rm End} (\mathcal{H})$ and assume that the evolution of the system is of Markovian type. Then, the evolution is given by a semigroup of completely positive dynamical maps $\Phi^L_t: \D \rightarrow \D$, for $t\geq 0$. The generator $L$ of the semigroup, called the Kossakowski-Lindblad operator, is an $\mathbb{R}$-linear operator acting on ${\rm End} (\mathcal{H})$. The Markovian evolution is determined by the differential equation
\begin{equation}
\label{KLDiff}
\frac{d}{dt} \rho(t) = L(\rho(t)), \quad \rho(t) = \Phi^L_t(\rho_0),
\end{equation}
where the expression of the Kossakowski-Lindblad operator is
\begin{equation}
\label{KLEq}
L(\rho) = -i[H,\rho] + \dfrac{1}{2} \sum_{j=1}^{n^2-1} ([V_j \rho, V_j^*] + [V_j, \rho V_j^*]) = -i[H,\rho] - \dfrac{1}{2} \sum_{j=1}^{n^2-1} [V_j^* V_j, \rho]_+ + \sum_{j=1}^{n^2-1} V_j \rho V_j^*,
\end{equation}
with $\rho \in \D$, $H^* = H$, ${\rm tr} (H) = 0$, ${\rm tr} (V_j) = 0$ and ${\rm tr}(V_j^* V_k) = 0$ if $j\neq k$, for $j,k = 1,2,\ldots, n^2-1$. For a given evolution, $H$ is uniquely determined by the trace restriction.
\end{theorem}

\begin{corollary}
Under Markovian evolution, rank of the state of a quantum system may change either at the initial time or at infinite time. It cannot change under finite time evolution.
\end{corollary}

Consider the stratification of the set of states in terms of the rank. According to Propositions \ref{PAStrat} and \ref{DAStrat}, Hamiltonian and gradient vector fields are tangent to each stratum. However, the rank of states is not preserved by generic Markovian dynamics. A quantum open system initially in a pure state and subjected to Markovian dynamics will in general evolve into a mixed state. A geometrical characterisation of Markovian dynamics requires vector fields with components transverse to the strata of $\D$, which therefore cannot be a combination of just Hamiltonian and gradient vector fields.

As seen in \cite{Grabowski2006}, the largest subgroup in the semigroup
of transformations acting on $\D$ is the general linear group, with
action given by $(A,\rho) \mapsto \frac{A \rho A^*}{{\rm tr}(A \rho
  A^*)}$. As the joint distribution of Hamiltonian and gradient vector
fields generate the action of the general linear group, the vector
field describing generic Markovian dynamics cannot belong to such a
distribution.  

An easy way to obtain new vector fields on the set of states is to
consider linear transformations in the whole manifold. Due to the
isomorphism $\O^* \cong T_\xi \O^*$, any $\mathbb{R}$-linear
transformation $T: \O^* \to \O^*$ defines a vector field
$\widehat{Z}_T \in \mathfrak{X}(\O^*)$
\cite[p. 108]{Carinena2015book}, whose action on any linear function
$f \in C^\infty (\O^*)$ is 
\begin{equation}
\widehat{Z}_T (f) (\xi) = f (T(\xi)), \quad \xi \in \O^*.
\end{equation}
As in the case of tensor fields, presented in previous sections, these vector field have to be slightly modified in order to fit in the geometry of the set of quantum states. In fact, observe that the vector field $\widehat{Z}_T$ defines a linear actions on $\O^*$, generated by $T$. This linearity is lost on the set of states, much in the same way as functions associated to observables are no longer linear functions when considered on $\D$.

\begin{proposition}
\label{propZT}
An $\mathbb{R}$-linear transformation $T: \O^* \to \O^*$ defines a vector field $Z_T \in \mathfrak{X} (\D)$ whose action on expectation value functions is
\begin{equation}
\label{eqZT}
Z_T (e_a) (\rho) = e_{T^\sharp(a)} (\rho) - e_{T^\sharp (I)} (\rho) e_a (\rho), \quad 
\rho \in \D, \quad
a \in \O,
\end{equation}
where $T^\sharp: \O \to \O$ denotes the dual map of $T$, defined as $T(\xi) (a) = \xi (T^\sharp (a))$.
\end{proposition}
\begin{proof}
In order to obtain this expression for $Z_T$, simply apply the vector field $\widehat{Z}_T$ on the larger manifold $\O^*$ on the pull-back of expectation value functions:
\begin{equation}
\label{ZTea}
\widehat{Z} (e_a) (\xi) = \frac{1}{f_I(\xi)} \ \widehat{Z} (f_a) (\xi) - \frac{f_a (\xi)}{f_I(\xi)^2} \ \widehat{Z} (f_I) (\xi) = e_{T^\sharp(a)} (\xi) - e_{T^\sharp (I)} (\xi) e_a (\xi), \quad
\xi \in \O^*, \ a \in \O,
\end{equation}
This expression {\color{black} corresponds again to a function on $\D$,
  and thus it defines} the proposed vector field. 
\end{proof}

\begin{remark}
It is important to notice the role of linearity in the geometric
description of states. A generic linear vector field on the larger
manifold $\O^*$ does not preserve the normalisation of states, and
thus cannot be restricted to $\D$. It is possible to consider its
action on relevant functions, i.e. on the pull-back to $\O^*$ of
expectation value functions. Thus, vector fields are obtained on $\D$
which by definition preserve the normalisation. Linearity, however, is
lost in \eqref{eqZT}. Observe that, if the defining transformation $T$
preserves the set of states, then $T^\sharp (I) = I$. This is for
example the case of Hamiltonian vector fields. Gradient vector fields,
on the contrary, are not linear on $\D$, as seen in the example of a
2-level system.
Notice that infinitesimal generators of unitary transformations
project onto linear vector fields while gradient vector field do not.
\end{remark}

Consider the action of the general linear group on the set of states
$\D$. This action, on $\O^*$, is generated by Hamiltonian and gradient
vector fields. Thus, by either \eqref{defGradHam} or Proposition
\ref{propZT}, the action on $\D$ is generated by vector fields of the
form 
\begin{equation}
W = c_H \widetilde{X}_H + c_F \widetilde{Y}_F , \quad
	H, F \in \mathcal{O}, \quad
	c_H, c_F \in \mathbb{R}.
\end{equation}
Using \eqref{LambdaFieldD} and the definition of gradient and Hamiltonian vector fields, the action of the vector field on an expectation value function is given by
\begin{equation}
W (e_a) (\rho) = c_H e_{\Ll H,a \Lr} (\rho) + c_F e_{F \odot a} (\rho) - c_F e_F (\rho) e_a (\rho), \quad \rho \in \D.
\end{equation}
{\color{black} Thus, the action of the general linear group on $\D$ is generated by
vector fields which act generally in a non-linear way on
$\D$. Linearity is obtained only for $c_F = 0$. A natural question
should be if there exists a more general scheme for which
linearity is recovered. As seen next, the answer is affirmative and it
corresponds to the generator of semigroups, as in the case of the
vector field associated to the Kossakowski-Lindblad 
equation.
}
Consider the application of Proposition \ref{propZT} to the particular case of the Kraus map. Let $V_1, \ldots, V_r$, with $r$ a natural number between 1 and $n^2-1$, be endomorphisms on the Hilbert space, and let us define the Kraus map $K$ as the $\mathbb{R}$-linear transformation on $\O^*$ given by
\begin{equation}
\label{KOp}
	K(\xi) = \sum_{j=1}^r V_j \xi V_j^*, \quad \xi \in \O^*.
\end{equation}

The properties of the Kraus maps have been studied to a great extent \cite{Grabowski2006}. According to the previous proposition, it defines a vector field
\begin{equation}
	Z_K(e_a) (\rho)
	= e_a \left(\sum_{j=1}^r V_j \rho V_j^* \right) - e_V(\rho) e_a (\rho)
	, \quad
	V = \sum_{j=1}^r V_j^* V_j \in \O, \quad
	\rho \in \D.
\end{equation}
The integral curves of this vector field preserve positivity. However,
the rank is not preserved, and therefore this vector field cannot be a
linear combination of Hamiltonian and gradient vector fields.

Now, let us consider the vector field $W'$ which is a linear combination of three different vector fields with real coefficients,
\begin{equation}
\label{combLin3}
	W' = c_H \widetilde{X}_H + c_F \widetilde{Y}_F + c_K Z_K, \quad
	H, F \in \mathcal{O}, \quad
	c_H, c_F, c_K \in \mathbb{R}.
\end{equation}
As before, the action of $W'$ on expectation value functions is
\begin{equation}
W' (e_a) (\rho) = c_H e_{\Ll H,a \Lr} (\rho) + c_F e_{F \odot a} (\rho) + c_K e_{a} (K(\rho)) - c_F e_F (\rho) e_a (\rho) - c_K e_V(\rho) e_a (\rho), \quad
\rho \in \D.
\end{equation}

The non-linearity of this vector field is now due to the last two terms. Linearity can be regained in this case by imposing a fine tuning between the last two vector fields:
\begin{equation}
	F = - \frac{c_K}{c_F} V = - \frac{c_K}{c_F} \sum_{j=1}^r V_j^* V_j.
\end{equation}

Taking $c_H = c_F = c_K = 1$, $F=-V$ in the definition \eqref{combLin3}, the resulting vector field is the generator of a semigroup of $\mathbb{R}$-linear transformations on the space of states. The associated differential equation for the integral curves of such a vector field is precisely the Kossakowski-Lindblad equation.

\begin{theorem}
\label{thmZL}
Let $H$ be an observable and let $K$ be the Kraus operator defined as in \eqref{KOp}. The vector field $Z_{L}$ on $\D$ defined as
\begin{equation}
\label{ZL}
Z_L = \widetilde{X}_H - \widetilde{Y}_V + Z_K, \quad V =\sum_{j=1}^{n^2-1} V_j^* V_j.
\end{equation}
is {\color{black} linear and } such that the system of differential
equations for its integral curves is given by the Kossakowski-Lindblad
equation \eqref{KLEq}. 
\end{theorem}

Notice that the notation is consistent. $Z_L$ is the vector field
associated to the $\mathbb{R}$-linear transformation $L: \O^* \to
\O^*$ in \eqref{KLEq}. The action of $Z_L$ on an expectation value
function is 
\begin{equation}
Z_L (e_a) (\rho) = e_{\Ll H,a \Lr} (\rho) - e_{V \odot a} (\rho) + e_{a} (K(\rho)) = e_{\Ll H,a \Lr} (\rho) - e_{V \odot a} (\rho) + e_{K^\sharp (a)} (\rho), \quad \rho \in \D,
\end{equation}
where the dual map $K^\sharp: \O \to \O$ is defined as in Proposition
\ref{propZT}.{\color{black} Notice that $Z_L (e_a) (\rho) $ is again an expectation
value function.
}
{\color{black}
Observe the relation between the Kossakowski-Lindblad vector field
$Z_L$ and the vector field $Z_{KM}$ for the Kaufman-Morrison
dissipation introduced in \eqref{eqKM}. Clearly, the lack of a Kraus
term in $Z_{KM}$ proves that in general this is not a Markovian
evolution. This could be directly proved by checking that $Z_{KM}$ is
not a linear vector field on $\D$. Also, unlike true Markovian
dynamics, the Kaufman-Morrison dissipation preserves the stratum of
pure states.
The present approach identifies Hamiltonian and gradient vector fields
which are related among them and are infinitesimal generators of the
maximal group of trace-preserving completely positive maps. It is  the
group of all invertible trace preserving completely positive maps
contained in the  semigroup of such maps. 

}

\begin{remark}
The last two vector fields in the decomposition of $Z_L$ given in
\eqref{ZL} are not independent. The gradient vector field $Y_V$ is
uniquely determined for each possible vector field $Z_K$. Such
relation follows, as indicated before, in order to obtain
$\mathbb{R}$-linear transformations in the space of states. The
resulting vector field is well defined in the whole set of states
$\D$, but is not in general tangent to each stratum of the set. Recall
that, as proved in \cite{Grabowski2005}{\color{black}, smooth curves are completely
contained in  the stratum they belong to}. Thus, change from one stratum to another can
occur either at initial time or at the limit of the evolution; for
finite time, Markovian evolution preserves the stratification of the
manifold. 
\end{remark}

\section{Open systems and dynamics on tensor fields}
\label{secContraction}

While the usual matrix mechanics describes the evolution of observables, the geometric formalism is more flexible. Given a vector field, as $Z_L$, one can obtain the evolution of any function or tensor field, for instance, functions as concurrence, purity or entropies. Future works will deal with the description of these aspects. 
{\color{black}
This section will focus instead on the behaviour of the geometric structures on the set of states $\D$, in particular tensor fields $\Lambda_\D$ and $R_\D$ presented in Theorem \ref{thTF}, under Markovian evolution. This can be done thanks to the geometric characterisation of Kossakowski-Lindblad equation by vector field $Z_L$ in Theorem \ref{thmZL}.
}

\begin{definition}
Let $M$ be a manifold. Let $X$ be a vector field on the manifold, whose flow is a semigroup of transformations $\{\phi^X_t: \D \to \D, t \geq 0\}$. The Lie derivative of a contravariant tensor field $T$ on $M$ with respect to the vector field $X$ is defined as
\begin{equation}
\mathcal{L}_{X} T = \lim_{t\rightarrow 0^+} \frac{1}{t} \left( T - \phi^X_{t*} T \right),
\end{equation}
\end{definition}

\begin{proposition}
\label{PropEvolTens}
Let $\D$ be the space of states of a quantum open system with a Markovian evolution. If $Z_L$ is the Kossakowski-Lindblad vector field, $\Phi_t^L$ its flow and $\Lambda_{\D,t}$ and $R_{\D,t}$ denote the families of tensor fields defined by {\color{black} the action of the flow on tensor fields $\Lambda_\D$ and $R_\D$:}
\begin{equation}
\begin{aligned}
\Lambda_{\D,t}
{\color{black} := \Phi^L_{t*} \Lambda_\D}
 = & e^{-t\mathcal{L}_{Z_L}} \Lambda_\D = \Lambda_\D - t \mathcal{L}_{Z_L} \Lambda_\D + \frac{t^2}{2!} (\mathcal{L}_{Z_L})^2 \Lambda_\D - \cdots, \\
R_{\D,t}
{\color{black} := \Phi^L_{t*} R_\D}
 = & e^{-t\mathcal{L}_{Z_L}} R_\D = R_\D - t \mathcal{L}_{Z_L} R_\D + \frac{t^2}{2!} (\mathcal{L}_{Z_L})^2 R_\D - \cdots,
\end{aligned}
\qquad t \geq 0,
\end{equation}
there exists a Lie-Jordan algebra of expectation value functions with composition laws defined by
\begin{equation}
\label{compLawt}
\{e_a, e_b\}_t (\rho) = \Lambda_{\D,t} (\d e_a, \d e_b) (\rho), \quad
(e_a, e_b)_t (\rho) = R_{\D,t} (\d e_a, \d e_b) (\rho) + e_a (\rho) e_b (\rho), \quad
\rho \in \D.
\end{equation}
with $a,b \in \O$. All the resulting algebras are isomorphic for any finite time $t\geq0$.
\end{proposition}
\begin{proof}
The transformations $\Phi^L_t$ defining the flow of the vector field $Z_L$ is invertible for any finite $t$. Therefore, tensorial properties are preserved, as it is a point transformation. For any finite $t>0$, the Lie-Jordan algebra of expectation value functions is isomorphic to the initial one. Hence all these algebras are isomorphic.
\end{proof}

{\color{black}
Families of tensor fields $\Lambda_{\D,t}$ and $R_{\D,t}$ have a huge relevance in the characterisation of the dynamics of quantum systems. Theorem \ref{propLJE} shows the relation of initial contravariant tensor fields $\Lambda_\D$ and $R_\D$ with the Lie-Jordan algebra of observables of the system. The tensor fields, however, are not in general constant along the evolution generated by the Kossakowski-Lindblad vector field $Z_L$. As a consequence, for every $t\geq 0$ there exists a different pair of tensor fields $\Lambda_{\D,t}$ and $R_{\D,t}$, which in turn define different Lie-Jordan algebras of expectation value functions by \eqref{compLawt}. Proposition \ref{PropEvolTens} shows that all these algebras are still isomorphic for every finite time.
}

The families $\Lambda_{\D,t}$ and $R_{\D,t}$ may or may not have asymptotic limits when $t \rightarrow \infty$. If they exist, let them be denoted
\begin{equation}
\Lambda_{\D,\infty} = \lim_{t \rightarrow \infty} \Lambda_{\D,t}, \quad
R_{\D,\infty} = \lim_{t \rightarrow \infty} R_{\D,t}.
\end{equation}
The interest of these limits rests on the algebra structure in the space of functions that they define.
{\color{black}
Unlike in the case of finite time, the new algebra of expectation value functions obtained in the asymptotic limit is not in general isomorphic to the initial one. Therefore, the evolution of an open system may define as an asymptotic limit a new product on the space of observables different from the initial one. This phenomenon is known as a contraction of the algebra. The idea of contractions was first introduced by Segal \cite{Segal1951} and also by In\"on\"u and Wigner \cite{Inonu1953}, in relation with the study of the classical limit of relativistic systems. The  tools developed by them proved to be useful in the study of Lie algebra. Contractions of Lie algebras have been deeply studied, specially by Weimar-Woods \cite{Weimar-Woods1991, Weimar-Woods1991a, Weimar-Woods2000, Weimar-Woods2006}.}

\begin{theorem}
\label{algEVF}
Suppose that the limits $\Lambda_{\D,\infty}$ and $R_{\D,\infty}$ of the families presented in Proposition \ref{PropEvolTens} do exit. Then, the set of expectation value functions $\E_\D$ on $\D$ is a Lie-Jordan algebra with respect to the products $\{\cdot, \cdot\}_\infty$ and $(\cdot, \cdot)_\infty$ defined as
\begin{equation}
\label{limProd}
\begin{aligned}
\{e_a, e_b\}_\infty (\rho) & = \Lambda_{\D,\infty} (\d e_a, \d e_b) (\rho), \\ 
(e_a, e_b)_\infty (\rho) & = R_{\D,\infty} (\d e_a, \d e_b) (\rho) + e_a (\rho) e_b (\rho),
\end{aligned} \qquad \rho \in \D.
\end{equation}
This algebra gives rise to an associative complex algebra with respect to the product
\begin{equation}
e_a *_\infty e_b = (e_a, e_b)_\infty + i \{e_a, e_b\}_\infty, \quad a,b \in \O.
\end{equation}
\end{theorem}
\begin{proof}
The new algebra is by definition a contraction of the initial algebra of expectation value functions (which was isomorphic to $\O$). As proven in \cite{Carinena2001}, algebraic properties depending on universal qualifiers (skew-symmetry, Jacobi identity, etc.) are preserved by the contraction procedure. Thus, the contracted algebra satisfies the same properties as the initial one, i.e. it is a Lie-Jordan algebra. It also has a unit, $e_I = 1$, as it is preserved by the contraction. Finally, complex associative algebras can always obtained from real Lie-Jordan algebras, as is the case.
\end{proof}

{\color{black}
It is remarkable that the contracted algebra of observables is still a Lie-Jordan algebra, altough the composition laws are different from those of the initial algebra.
}
Works by some of us describe the general theory of contractions for any type of algebras \cite{Carinena2000, Carinena2001, Carinena2004}. The particularisation to contractions of Jordan, Lie-Jordan or associative algebras, however, has not been described yet. There exist works dealing with applications of contractions, such as \cite{Alipour2015, Chruscinski2012}, where the contraction of algebras of observables of a quantum system is studied in an algebraic setting.


From a more algebraic point of view, the contracted Poisson and symmetric brackets no longer represent the commutator or anti-commutator of observables. In fact, the new products define a new pair of operations $\Ll \cdot, \cdot \Lr_\infty$ and $\odot_\infty$ on the set of observables,
\begin{equation}
\{e_a, e_b\}_\infty = e_{\Ll a,b\Lr_\infty}, \quad
(e_a, e_b)_\infty = e_{a \odot_\infty b}, \quad
a,b \in \O.
\end{equation}
which are different from the initial ones. Thus, a contraction of the algebra of observables is obtained. These contractions have been previously studied in \cite{Alipour2015, Chruscinski2012,Ibort2016}. Due to the nature of the contraction procedure, some non-commuting observables $[a,b] \neq 0$ in the initial algebra may satisfy $[a,b]_\infty = 0$. Similarly, the non-associativity of the Jordan product may disappear, obtaining $a \odot_\infty b = ab$. Thus, the contraction of the algebra is connected with the transition from quantum to classical observables. The physical implications of the contraction procedure is a promising topic that will be discussed in future works.

In order to illustrate the description of contractions of Lie-Jordan algebras, several examples of open systems will be considered. Under Kossakowski-Lindblad evolutions, the families of tensor fields will be determined and, if the asymptotic limit exists, the limit algebras of the evolution will be found.
{\color{black}
The first examples correspond to the phase damping and the dissipation of 2-level systems. For completeness, a short analysis of Markovian evolutions of three-level systems are also presented.
}

{\color{black}
\subsection{Phase damping of open 2-level systems}
}

Let us recover the coordinate expressions for a two-level system presented in Section \ref{section2levels}. A basis of the space of observables was given by the Pauli matrices and the identity matrix, which determined the coordinate expressions (\ref{LambdaR}) for the contravariant tensor fields $\Lambda_\D$ and $R_\D$.  

{\color{black} As a first practical example in the analysis of contractions of tensor fields, }
let us consider the phase damping of a qubit, given by the following Kossakowski-Lindblad operator \cite{Chruscinski2012, Alipour2015,Ibort2016}:
\begin{equation}
\label{LPhaseDamp}
L (\rho) = - \gamma (\rho - \sigma_3 \rho \sigma_3), \quad \rho \in \D.
\end{equation}
In order to obtain the associated vector field $Z_L$ on $\D$, consider the basis $\{\sigma_\mu\}_{\mu=0}^3$ for $\O \cong {\rm Herm} (2)$ described in Section \ref{section2levels}. As $L$ is a self-adjoint operator on matrices, i.e. $L = L^\sharp$, we find by direct computation that
\begin{equation*}
L^\sharp (\sigma_1) = -2\gamma \sigma_1, \quad
L^\sharp (\sigma_2) = -2\gamma \sigma_2, \quad
L^\sharp (\sigma_3) = 0, \quad
L^\sharp (I) = 0.
\end{equation*}
With this basis, the coordinate expression of the vector field $Z_L$ {\color{black} associated to the Kossakowski-Lindblad operator in \eqref{LPhaseDamp}}, computed directly by \eqref{eqZT}, is:
\begin{equation}
\label{ZL1}
Z_L = -2\gamma \left( x_1 \pdx{1} + x_2 \pdx{2} \right).
\end{equation}
The 2-level system has great advantages from a practical point of view. The Lie derivatives of $\Lambda_\D$ and $R_\D$ with respect to this vector field can be directly computed:
\begin{equation*}
\mathcal{L}_{Z_L} (\Lambda_\D) = 4 \gamma x_3 \pdx{1} \wedge \pdx{2}, \quad
\mathcal{L}_{Z_L} (R_\D) = 4 \gamma \pdx{1} \otimes \pdx{1} + 4 \gamma \pdx{2} \otimes \pdx{2}.
\end{equation*}
In order to compute the coordinate expressions of the families $\Lambda_{\D,t}$ and $R_{\D,t}$, consider simply the expansion given in Proposition \ref{PropEvolTens}. Thus, the resulting $t$-dependent tensor fields are
\begin{equation}
\label{LDRDt1}
\begin{aligned}
\Lambda_{\D,t} = & e^{-4\gamma t} x_3 \pdx{1} \wedge \pdx{2} + x_1 \pdx{2} \wedge \pdx{3} + x_2 \pdx{3} \wedge \pdx{1}, \\
R_{\D,t} = & e^{-4\gamma t} \left( \pdx{1}\otimes\pdx{1} + \pdx{2}\otimes\pdx{2} \right) + \pdx{3}\otimes\pdx{3} - x_j x_k \pdx{j} \otimes \pdx{k},
\end{aligned} \qquad t \geq 0.
\end{equation}

\begin{theorem}
\label{thmEx1}
There exist asymptotic limits $\Lambda_{\D,\infty}$ and $R_{\D,\infty}$ for the families of tensor fields given in \eqref{LDRDt1}, determined by the {\color{black} phase damping} evolution generated by the vector field \eqref{ZL1}. The limits are
\begin{equation}
\Lambda_{\D,\infty} = x_1 \pdx{2} \wedge \pdx{3} + x_2 \pdx{3} \wedge \pdx{1}, \quad
R_{\D,\infty} = \pdx{3}\otimes\pdx{3} - x_j x_k \pdx{j} \otimes \pdx{k}.	
\end{equation}
{\color{black} The products of smooth functions on $\D$ defined by}
\begin{equation}
\label{prodsInf1}
\{f,g\}_\infty = \Lambda_{\D, \infty} (\d f, \d g), \quad
(f,g)_\infty = R_{\D, \infty} (\d f, \d g) +fg.
\end{equation}
{\color{black}  are a Poisson bracket and  a symmetric product, respectively.}
\end{theorem}
\begin{proof}
The existence of the limits to \eqref{LDRDt1} is clear by direct inspection. Tensorial properties are preserved in the family, as in Theorem \ref{algEVF}. Thus the resulting tensor fields define the same type of products as the initial ones.
\end{proof}

It is immediate to check that the set of expectation value functions $\E_\D$ on $\D$ is a Lie-Jordan algebra with respect to the products $\{\cdot, \cdot\}_\infty$ and $(\cdot, \cdot)_\infty$. Recall from \eqref{expVF2lev} that the expectations value function associated to an observable $a= a^\mu \sigma_\mu$ is given by $e_a = e^0 + a^j x_j$. As the products are $\mathbb{R}$-linear, it is enough to describe the products of constant and linear functions. The unit function satisfies $\{f,1\}_\infty = 0$ and $(f,1)_\infty = f$ for any smooth function $f$ on $\D$. Regarding linear functions, they satisfy the following products:
\begin{equation}
\label{eqAlgEVFEx1}
\begin{aligned}
	& \{x_1, x_3\}_\infty = -x_2, \quad
	\{x_2, x_3\}_\infty = x_1, \quad
	\{x_1, x_2\}_\infty = 0, \\
	& (x_1, x_1)_\infty = (x_2, x_2)_\infty = 0, \quad
	(x_3, x_3)_\infty = 1.
\end{aligned}
\end{equation}
and the rest of the products vanish identically. It is immediate to check that these products define a Lie-Jordan algebra.

Similarly, the $*_\infty$-product of functions introduced in Theorem \ref{algEVF} can be computed. The constant unit function acts as the unit element, as $e_a *_\infty 1 = 1 *_\infty e_a = e_a$ for any expectation value function. The product of linear functions is
\begin{equation}
\label{starPrLin1}
\begin{aligned}
x_1 *_\infty x_1 & = 0, & x_1 *_\infty x_2 & = 0, & x_1 *_\infty x_3 & = -i x_2, \\
x_2 *_\infty x_1 & = 0, & x_2 *_\infty x_2 & = 0, & x_2 *_\infty x_3 & = i x_1, \\
x_3 *_\infty x_1 & = ix_2, & x_3 *_\infty x_2 & = -ix_1, & x_3 *_\infty x_3 & = 1.
\end{aligned}
\end{equation}
It can be check by direct computation that the $*_\infty$-product is associative.

\begin{remark}
The phase damping defines a contraction of the algebra of expectation values. That is, starting from the algebra given by \eqref{algEVF2Lev}, the evolution defines a transformation of the products. In the asymptotic limit, the algebra described in \eqref{eqAlgEVFEx1} is obtained. This new algebra is not isomorphic to the initial one, however it is still a Lie-Jordan algebra, and it gives rise to a complex associative product. Notice that the contraction defines new products over the same linear space: the functions are not modified, and they are still expectation value functions on $\D$.
\end{remark}

For the sake of completeness, observe that Lie algebra $(\E_\D, \{\cdot, \cdot\}_\infty)$ is isomorphic to the Lie algebra on the plane. This result is in agreement with previous works \cite{Chruscinski2012, Alipour2015,Ibort2016}, which obtain similar results from an algebraic computation. The tensorial description presents the advantage of dealing directly with the algebraic structures codified in terms of tensor fields. As seen in the example, it is not necessary to compute the evolution of expectation value functions in order to obtain the results. As fewer objects are to be dealt with, a possible generalisation to more abstract settings can thus be more easily achieved in the tensorial description.

{\color{black}
\subsection{Dissipation of open 2-level systems}
}

Other examples of Markovian evolution can be considered. The dynamics that describes the dissipation of a qubit is given in \cite{Baumgartner2008b}. Section 5.1 of that work presents the following Kossakowski-Lindblad operator:
\begin{equation}
\label{KLOpBaumg}
	L (\rho) = D_{J_+} (\rho) + D_{J_-} (\rho), \quad D_J(\rho) = J \rho J^* - \frac 12 (J^* J \rho + \rho J^* J),
\end{equation}
where $J_+$ and $J_-$ are the ladder operators for a qubit,
\begin{equation}
	J_+ = \begin{pmatrix} 0 & 1 \\ 0 & 0 \end{pmatrix}, \quad
	J_- = \begin{pmatrix} 0 & 0 \\ 1 & 0 \end{pmatrix}.
\end{equation}
Once again with \eqref{eqZT}, the vector field $Z_L$ associated to 
{\color{black}
the Kossakowski-Lindblad operator in \eqref{KLOpBaumg} }
is
\begin{equation}
\label{ZL2}
	Z_L = -x_1 \pdx{1} -x_2 \pdx{2} -2x_3 \pdx{3},
\end{equation}
which gives the following families of contravariant tensor fields:
\begin{equation}
\label{LDRDt2}
\begin{aligned}
	\Lambda_{\D,t} = & x_3 \pdx{1} \wedge \pdx{2} + e^{-2t} x_1 \pdx{2} \wedge \pdx{3} + e^{-2t} x_2 \pdx{3} \wedge \pdx{1}, \\
	R_{\D,t} = & e^{-2t}\left( \pdx{1}\otimes\pdx{1} + \pdx{2}\otimes\pdx{2} \right) + e^{-4t} \pdx{3}\otimes\pdx{3} - x_j x_k \pdx{j} \otimes \pdx{k}.
\end{aligned}
\end{equation}
This evolution also has an asymptotic limit for the tensor fields. A new contraction of the algebra of expectation functions, hence of the algebra of quantum observables, is obtained.

\begin{theorem}
\label{thmEx2}
There exist asymptotic limits $\Lambda_{\D,\infty}$ and $R_{\D,\infty}$ for the families of tensor field given in \eqref{LDRDt2}, determined by the Markovian evolution generated by the vector field \eqref{ZL2}. The limits are
\begin{equation}
\Lambda_{\D,\infty} = x_3 \pdx{1} \wedge \pdx{2}, \quad
R_{\D,\infty} = - x_j x_k \pdx{j} \otimes \pdx{k}.	
\end{equation}
The products $\{f,g\}_\infty = \Lambda_{\D, \infty} (\d f, \d g)$ and $(f,g)_\infty = R_{\D, \infty} (\d f, \d g) +fg$ of smooth functions are respectively a Poisson bracket and a symmetric product
\end{theorem}

Again, the set of expectation value functions $\E_\D$ on $\D$ is a Lie-Jordan algebra with respect to these products, which can be checked by computing the product of $x_j$ functions:
\begin{equation}
\label{prodsEx2}
\{x_1, x_2\}_\infty = x_3, \quad \{x_1, x_3\}_\infty = \{x_2, x_3\}_\infty = 0; \qquad
(x_j, x_k) = 0, \quad j,k = 1,2,3.
\end{equation}
and the rest of product vanish identically. Regarding the product $f *_\infty g = (f,g)_\infty + i \{f,g\}_\infty$, the only non-zero products of $x_j$ functions are
\begin{equation}
x_1 *_\infty x_2 = - (x_2 *_\infty x_1) = i x_3.
\end{equation}
Thus, the $*_\infty$-product is associative.

\begin{remark}
It can be concluded from \eqref{prodsEx2} that contracted Lie algebra $(\E_\D, \{\cdot, \cdot\}_\infty)$ is isomorphic to the Heisenberg algebra. As proved in \cite{Weimar-Woods1991}, the only non-trivial contractions of the $\mathfrak{su}(2)$ Lie algebra are the Euclidean algebra and the Heisenberg algebra. It is thus possible to describe all the possible contractions of this algebra by means of Markovian evolution of the corresponding quantum system. Also, with our approach, the Jordan algebra is also contracted, thus obtaining all the non-trivial contractions of the Lie-Jordan algebra of observables of a two-level quantum system.
\end{remark}

\subsection{Open 3-levels systems}
\label{sec:open-3-levels}
The manifold of states of a three-level system presents a richer structure than that of two-level systems. While the later is composed of two strata, manifolds of states of three-level systems are decomposed in three strata, two of them composing the boundary. It is therefore of interest to consider evolution of a three-level system. In particular, some examples of evolution of 3-level systems that produce contractions of algebras will be studied

Let us consider the isomorphism $\mathcal{O} \cong {\rm Herm}(\mathbb{C}^3)$. A basis of the algebra is given by the Gell-Mann matrices and the identity matrix \cite{Clemente-Gallardo2007}. The expectation value functions of the Gell-Mann matrices are the coordinate functions $x_1, \ldots, x_8$ on $\D$. For completeness, the products of these functions are presented here. They are directly computed by their definition and the products of Gell-Mann matrices.
\begin{proposition}
\label{propProds3L}
The Poisson brackets of the coordinate functions on the set of states of a 3-level quantum system are
\begin{equation}
\begin{aligned}
	& 
	\{x_1, x_2\} = x_3, \,
	\{x_1, x_3\} = - x_2, \,
	\{x_1, x_4\} = \frac 12 x_7, \,
	\{x_1, x_5\} = -\frac 12 x_6, \,
	\{x_1, x_6\} = \frac 12 x_5, \\
	&
	\{x_1, x_7\} = -\frac 12 x_4, \,
	\{x_2, x_3\} = x_1, \,
	\{x_2, x_4\} = \frac 12 x_6, \,
	\{x_2, x_5\} = \frac 12 x_7, \,
	\{x_2, x_6\} = -\frac 12 x_4, \\
	&
	\{x_2, x_7\} = -\frac 12 x_5, \,
	\{x_3, x_4\} = \frac 12 x_5, \,
	\{x_3, x_5\} = -\frac 12 x_4, \,
	\{x_3, x_6\} = -\frac 12 x_7, \,
	\{x_3, x_7\} = \frac 12 x_6, \\
	&
	\{x_4, x_5\} = \frac 12 (x_3 + \sqrt{3} x_8), \,
	\{x_4, x_6\} = \frac 12 x_2, \,
	\{x_4, x_7\} = \frac 12 x_1, \,
	\{x_4, x_8\} = -\frac{\sqrt{3}}{2} x_5, \\
	&
	\{x_5, x_6\} = -\frac 12 x_1, \,
	\{x_5, x_7\} = \frac 12 x_2, \,
	\{x_5, x_8\} = \frac{\sqrt{3}}{2} x_4, \,
	\{x_6, x_7\} = -\frac 12 (x_3 - \sqrt{3} x_8), \\
	&
	\{x_6, x_8\} = -\frac{\sqrt{3}}{2} x_7, \,
	\{x_7, x_8\} = \frac{\sqrt{3}}{2} x_6,
\end{aligned}
\end{equation}
and the non-listed products vanish identically.
The Jordan brackets of the coordinate functions on the set of states of a 3-level quantum system are
\begin{equation}
\begin{aligned}
&
(x_1, x_1) = \frac 4 3 + \frac{2}{\sqrt{3}} x_8, \,
(x_1, x_4) = x_6, \,
(x_1, x_5) = x_7, \,
(x_1, x_6) = x_4, \,
(x_1, x_7) = x_5, \\
&
(x_1, x_8) = \frac 23 x_1 (\sqrt{3}-3 x_8)+2 x_1 x_8, \,
(x_2, x_2) = \frac 43 + \frac{2}{\sqrt{3}} x_8, \,
(x_2, x_4) = -x_7, \\
&
(x_2, x_5) = x_6, \,
(x_2, x_6) = x_5, \,
(x_2, x_7) = -x_4, \,
(x_2, x_8) = \frac 23 x_2 (\sqrt{3}-3 x_8)+2 x_2 x_8, \\
&
(x_3, x_3) = \frac 43 + \frac{2}{\sqrt{3}} x_8, \,
(x_3, x_4) = (1-2 x_3) x_4+2 x_3 x_4, \,
(x_3, x_5) = (1-2 x_3) x_5+2 x_3 x_5, \\
&
(x_3, x_6) = 2 x_3 x_6-(1+2 x_3) x_6, \,
(x_3, x_7) = 2 x_3 x_7-(1+2 x_3) x_7, \\
&
(x_3, x_8) = \frac 23 x_3 (\sqrt{3}-3 x_8)+2 x_3 x_8, \,
(x_4, x_4) = \frac 43 + x_3- \frac{1}{\sqrt{3}} x_8, \,
(x_4, x_6) = x_1, \\
&
(x_4, x_7) = -x_2, \,
(x_4, x_8) = 2 x_4 x_8-\frac 13 x_4 (\sqrt{3}+6 x_8), \,
(x_5, x_5) = \frac 43+x_3-\frac{1}{\sqrt{3}} x_8, \\
&
(x_5, x_6) = x_2, \,
(x_5, x_7) = x_1, \,
(x_5, x_8) = 2 x_5 x_8- \frac 13 x_5 (\sqrt{3}+6 x_8), \\
&
(x_6, x_6) = \frac 43 -x_3- \frac{1}{\sqrt{3}} x_8, \,
(x_6, x_8) = 2 x_6 x_8 - \frac 13 x_6 (\sqrt{3}+6 x_8), \\
&
(x_7, x_7) = \frac 43 -x_3- \frac{1}{\sqrt{3}} x_8, \quad
(x_7, x_8) = 2 x_7 x_8- \frac{1}{3} x_7 (\sqrt{3}+6 x_8), \\
&
(x_8, x_8) = 2 x_8^2- \frac 23 (-2+\sqrt{3} x_8+3 x_8^2),
\end{aligned}
\end{equation}
the non-listed products being identically zero.
\end{proposition}

When dynamics is considered, the computations needed to describe the evolution of the associated tensor fields are similar to those of the 2-level system, however lengthy due to the higher dimension of the space of states. Thus, computations will be skipped and only the interesting results will be presented.

The first example is a particular case of two different models presented in \cite{Chruscinski2012, Alipour2015,Ibort2016}. These models are valid for any number $d$ of levels; they will be later particularised to the case $d=3$. The first model is the decoherence for massive particles, given by
\begin{equation}
\label{defL1}
	L_M(\rho) = -\gamma [X, [X, \rho]], \quad \rho \in \D,
\end{equation}
where $X$ is the position operator. This model can be discretised by considering a finite number $d$ of positions $\vec{x}_m$ along a circle. The positions are given by
\begin{equation}
	\vec{x}_m = (\cos \phi_m, \sin \phi_m), \quad \phi_m = \frac{2\pi m}{d}, \quad m= 1,2, \ldots, d.
\end{equation}
Let $\{|m\rangle\}_{m=1}^d$ denote the basis of eigenstates of the position operator. In this basis,the Kossakowski-Lindblad operator $L_M$ takes the form
\begin{equation}
	L_M |m\rangle \langle n| = -\gamma \left|\vec{x}_m - \vec{x}_n\right| |m\rangle \langle n| = -4\gamma \sin^2 \left( \frac{\pi(m-n)}{d} \right) |m\rangle \langle n|,
\end{equation}
for $m,n = 1,2, \ldots, d$.

On the other hand, the pure decoherence model of a $d$-level system is given by
\begin{equation}
\label{defL2}
	L_P(\rho) = -\frac 1d \sum_{k=1}^{d-1} \gamma_k (\rho - U_k \rho U_k^*),
	\quad \gamma_k >0, \quad k = 1,2,\ldots,d-1, \quad \rho \in \D,
\end{equation}
where $U_k$ are the unitary operators given by
\begin{equation}
	U_k = \sum_{l=1}^{d-1} \lambda^{-k(l-1)} P_l, \quad \lambda = e^{\frac{2\pi i}{d}},
\end{equation}
and $P_l$ are the 1-dimensional projectors $|l\rangle \langle l|$.

Taking $d=3$ in both models, we obtain a Markovian evolution for the three-level system. Starting from the operators defined in \eqref{defL1} and \eqref{defL2}, it is immediate to obtain the corresponding vector fields on $\D$ by \eqref{eqZT}. Simple but lengthy computations, similar to those carried out in the previous section, allow us to obtain the evolutions of the tensor fields $\Lambda_\D$ and $R_\D$, or, equivalently, of the products of expectation value functions given in Proposition \ref{propProds3L}. It it thus immediate to prove that both evolutions define identical contractions of the algebra of expectation value functions, as summarised in the following Theorem.

\begin{theorem}
The evolutions of a 3-level system by either the decoherence model of massive particles or the pure decoherence model define new products on the asymptotic limit. The Poisson bracket of $x_j$ functions is
\begin{equation}
\begin{aligned}
	& \{x_1, x_3\}_\infty = - x_2, \,
	\{x_2, x_3\}_\infty = x_1, \\
	& \{x_4, x_3\}_\infty = -\frac 12 x_5, \,
	\{x_5, x_3\}_\infty = \frac 12 x_4, \,
	\{x_4, x_8\}_\infty = -\frac{\sqrt{3}}{2} x_5, \,
	\{x_5, x_8\}_\infty = \frac{\sqrt{3}}{2} x_4, \\
	& \{x_6, x_3\}_\infty = \frac 12 x_7, \,
	\{x_7, x_3\}_\infty = -\frac 12 x_6, \,
	\{x_6, x_8\}_\infty = -\frac{\sqrt{3}}{2} x_7,\,
	\{x_7, x_8\}_\infty = \frac{\sqrt{3}}{2} x_6,
\end{aligned}
\end{equation}
The Jordan bracket in the asymptotic limit is given by the following expressions
\begin{equation}
\begin{aligned}
	& (x_3 , x_3)_\infty = \frac 23 + \frac{1}{\sqrt{3}} x_8, \,
	(x_8 , x_8)_\infty = \frac 23 -\frac{1}{\sqrt{3}} x_8, \\
	& (x_1 , x_8)_\infty = \frac{1}{\sqrt{3}} x_1, \,
	(x_2 , x_8)_\infty = \frac{1}{\sqrt{3}} x_2, \,
	(x_3 , x_8)_\infty = \frac{1}{\sqrt{3}} x_3, \,
	(x_4 , x_8)_\infty = -\frac{1}{2\sqrt{3}} x_4, \\
	& (x_5 , x_8)_\infty = -\frac{1}{2\sqrt{3}} x_5, \,
	(x_6 , x_8)_\infty = -\frac{1}{2\sqrt{3}} x_6, \,
	(x_7 , x_8)_\infty = -\frac{1}{2\sqrt{3}} x_7, \\
	& (x_4 , x_3)_\infty = \frac 12 x_4, \,
	(x_5 , x_3)_\infty = \frac 12 x_5, \,
	(x_6 , x_3)_\infty = -\frac 12 x_6, \,
	(x_7 , x_3)_\infty = -\frac 12 x_7.
\end{aligned}
\end{equation}
The set $\E_\D$ of expectation value functions is a Lie-Jordan algebra with respect to these two products.
\end{theorem}
\begin{proof}
The theorem is proven by the same arguments than those given for two-level systems. It is immediate, although lengthy, that the products satisfy the axioms of Lie-Jordan algebras.
\end{proof}

For a 3-level system, we can also find evolutions which do not describe contractions of algebras. The model of the decay to a 2-level system presented in section 5.3 of \cite{Baumgartner2008b} gives the following dynamics for the system:
\begin{equation}
L (\rho) = D_{J_1} (\rho) + D_{J_2} (\rho), \quad D_J(\rho) = J \rho J^* + \frac 12 (J^* J \rho + \rho J^* J), \quad \rho \in \D,
\end{equation}
where $J_1$ and $J_2$ are the following operators:
\begin{equation}
	J_1 = \begin{pmatrix} 0 & 0 & 1 \\ 0 & 0 & 0 \\ 0 & 0 & 0 \end{pmatrix}, \quad
	J_2 = \begin{pmatrix} 0 & 0 & 1 \\ 0 & 0 & 1 \\ 0 & 0 & 0 \end{pmatrix}.
\end{equation}

The vector field associated to the Kossakowski-Lindblad equation is
\begin{equation}
\begin{aligned}
\label{ZLBaum3lev}
Z_L = & \frac 23 (1 - \sqrt{3}\,x_8) \pdx{1} + \frac 13 (1 - \sqrt{3}\,x_8) \pdx{3} -\frac 12 x_4 \pdx{4} -\frac 12 x_5 \pdx{5} \\
&-\frac 12 x_6 \pdx{6} -\frac 12 x_7 \pdx{7} + (\sqrt{3} - 3x_8) \pdx{8}.
\end{aligned}
\end{equation}
The families of contravariant tensor fields $\Lambda_{\D,t}$ and $R_{\D,t}$ have a divergence on the asymptotic limit, which means that no algebra contraction is obtained. In particular, the divergent products in the algebra of functions are
\begin{equation}
\begin{aligned}
	& \{x_1, x_2\}_t =
	x_3 + \left(e^{3 t} -1\right) \frac{1- \sqrt{3}\, x_8}{3}, \quad
	\{x_2, x_3\}_t = x_1+ 2\left(e^{3 t} -1\right) \frac{1-\sqrt{3}\, x_8}{3}, \\
	& (x_1, x_1)_t = 1 + (4e^{-3t} + 5e^{3t} -12) \frac{1-\sqrt{3}\, x_8}{9}, \\
	& (x_2, x_2)_t = 1 + (3e^{3t} -4) \frac{1-\sqrt{3}\, x_8}{3}, \\
	& (x_3, x_3)_t = 1 + (e^{-3t} + 8e^{3t} -12) \frac{1-\sqrt{3}\, x_8}{9}, \\
	& (x_1, x_3)_t = 2 (e^{-3t} - e^{3t}) \frac{1-\sqrt{3}\, x_8}{9}.
\end{aligned}
\end{equation}

\begin{proposition}
The evolution of a 3-level system defined by the vector field given in \eqref{ZLBaum3lev} does not define a contraction of the Lie-Jordan algebra of functions.
\end{proposition}

\begin{remark}
Let us consider the semigroup of transformations $\{\Phi_t^L\}$ determined by the vector field in \eqref{ZLBaum3lev}, and in particular its asymptotic limit $\Phi_\infty^L$. The image by this transformation of the whole set $\D$ is the subset of limit points of the dynamics. It can be proved that, in this particular case, the limit points of the dynamics are precisely the fixed points in which the vector field $Z_L$ vanishes. These fixed points can be computed from the expression of the vector field:
\begin{equation}
\Phi_\infty^L (\D) = \left\{ \rho_L = (x_1, \ldots, x_8) \in \D \mid x_4 = x_5 = x_6 = x_7 = 0, \, x_8 = \frac{1}{\sqrt{3}} \right\}.
\end{equation}
If the products given above are evaluated in limit points $\rho_L \in \Phi_\infty^L (\D)$, the divergence is cancelled:
\begin{equation*}
\{x_1, x_2\}_t (\rho_L) = x_3, \quad
\{x_2, x_3\}_t (\rho_L) = x_1, \quad
(x_1, x_1)_t (\rho_L) = 1, \quad \mbox{etc.}.
\end{equation*}
Therefore, one could define a Lie-Jordan algebra of functions on the set of limit points $\Phi_\infty^L (\D)$. Such algebra consists of those functions that are non-constant in the set, say $x_1, x_2$ and $x_3$. The products in the algebra, given by the expression above, are
\begin{equation}
\begin{aligned}
&\{x_1, x_2\}_L = x_3, \quad
\{x_2, x_3\}_L = x_1, \quad
\{x_3, x_1\}_L = x_2, \\
& (x_1, x_1)_L = (x_2, x_2)_L = (x_3, x_3)_L = 1,
\end{aligned}
\end{equation}
and the remaining products are identically vanishing.
This is the Lie-Jordan algebra of functions of a two-level system. Hence the procedure here sketched can describe the algebra of functions on the set of limit points of Markovian dynamics. It should be stressed that this is not a contraction of the algebra of functions, as the final algebra has lower dimension than the algebra of functions of the quantum system. The description of the algebra of functions on limit sets has interesting applications in control theory. Future works will deal with a generalisation of the contraction procedure, with applications to control of open quantum systems.
\end{remark}

\section{Conclusions}
\label{secConcl}

{\color{black}
To put our paper into perspective, let us try to syntetize what we
have done.
The present work describes in geometrical terms the Heisenberg picture
of quantum mechanics  using as primary carrier space the
  set of states.
We have started from the space of quantum states considered as a
closed convex set $\mathcal{D}$ of an affine space. We have considered
this set as a stratified differential manifold with a non-smooth boundary. 
We have endowed the set of states with two
contravariant $(0,2)$-tensor fields: a skew-symmetric one which defines
a Poisson tensor, not degenerate on pure states, and a symmetric one
which on pure states is also invertible.  On the set of pure states, which
defines one stratum of $\mathcal{D}$,  the two tensors are
compatible in the sense that they define a Kahler manifold.  Indeed,
being both tensors fields invertible, we can define a complex structure
by defining an endomorphism which maps a gradient vector field onto
the Hamiltonian vector field, both associated with the same function. 
This construction can
be done on each stratum, but outside the pure states the rank of the
resulting tensor will no longer be constant. 

Observables are those smooth functions  on $\mathcal{D}$ whose
associated Hamiltonian vector fields are "Killing vector fields " for
the symmetric tensor field. They form a vector space, which can be
endowed with a Lie-Jordan algebra structure by using the symmetric and
the skew-symmetric tensors. Furthermore, it can 
be extended to a $C^{*}$-algebra of complex functions (with real and
imaginary part being separately observables), if we combine both
tensors. When the resulting $C^{*}$-algebra is maximally 
non-commutative (see \cite{GrabMar2003} for details), the initial
space of states is the space  of states of a quantum system.
The Lie algebra of Hamiltonian and gradient vector fields associated
with observables close on the Lie algebra of
$\mathrm{SL}(N,\mathbb{C})$, i.e., the complexification of the Lie
algebra of the unitary group.
The strata of the stratification of $\mathcal{D}$ are the integral
leaves of the involutive distribution defined by putting together
Hamiltonian and gradient vector fields.

By using this geometrical formalism, we have also considered the
description of the markovian evolution of open quantum systems.
As the dynamics is  required to be the infinitesimal generator of a
one-parameter semigroup of completely positive maps, it turns out to be 
decomposable into the sum of three vector fields: one Hamiltonian, one
gradient and a Kraus vector field, which is related with the gradient
one in such a way that their sum defines a linear vector field.

This geometric characterisation of the Kossakowski-Lindblad equation
allows us to obtain more information on the dynamics it describes. In
particular, in some particular cases a contraction of the algebra of
observables of quantum systems is obtained; such contraction can
easily be described in the geometrical formalism in terms of families
of tensor fields. Some examples illustrate this new
feature of the dynamics. As indicated, a tensorial approach to the
study of contractions has some advantages in contrast with an
algebraic one. It is possible to design an algorithm that computes the
existence of contractions for any Markovian evolution and any
dimension of the quantum system. Even more importantly, the tensorial
description deals directly with the tensor fields describing the
structure. In contrast, an algebraic approach requires the
computation of the evolution of the elements in the algebra in order
to obtain the contraction. Therefore, the results presented here
greatly encourages the tensorial description of algebraic structures
in order to study and generalise the concept of contraction. 

Future works will deal with a deeper description of the contraction of
algebras of open systems and its applications to quantum control. An
extension of the formalism to infinite-dimensional systems, in
particular to the study of coherent states of the quantum harmonic
oscillator and the connection with the approach  by
  Bonet-Luz, Ohsawa and Tronci \cite{Bonet-Luz2015, 
  Bonet-Luz2016, Ohsawa2016}.  , will also be the topic of coming papers. 

Finally, we would also like to compare our approach with others aiming
at providing a geometric description of quantum systems, as the one
introduced by Mielnik several years ago  and discussed recently by
Bengtsson and Zyczkowski (\cite{Mielnik1968,
  Mielnik1969,Mielnik1974,Bengtsson2006}).  The main difference of our
approach with those lies in the differential geometric treatment:
while Mielnik's approach does consider an approach to quantum states
as points of a convex topological space, we consider a finite
dimensional situation where we can consider the space of states  $\mathcal{D}$
as a differential geometric manifold. 
Transformations considered by Mielnik must preserve the convex
structure of the space of states. Instead, when we define our tensors
as $\Lambda_{\mathcal{D}}$ or $R_{\mathcal{D}}$, and the corresponding
brackets on the set of expectation value functions, we are considering
tensor fields on $\mathcal{D}$, with respect to general
nonlinear transformations. This becomes relevant when considering
nonlinear functions on the quantum states, such as von Neumann entropy
or concurrence. The existence of the compatible pair of tensor fields on
$\mathcal{D}$, allows us, in contrast with Mielnik, to completely
identify the observable functions. The Lie group which they generate
also identifies the motion group as required by his approach. The
requirement of maximal noncommutativity removes possible "more
classical" aspects of the described system, as considered by Mielnik.
}

\section*{Acknowledgments}
This research has been financially supported by the following Spanish
grants:
MICINN grant
FIS2013-46159-C3-2-P, MINECO grant MTM2015-64166-C2-1, Santander-UCIIM ‘Cátedra
de Excelencia Program 2016/2017’, DGA grants 24/1 and B100/13, and MECD grant FPU13/01587.

\end{document}